\documentclass[10pt, conference, letterpaper]{IEEEtran}

\usepackage{framed}
\usepackage{hyperref}
\usepackage{tcolorbox}
\usepackage{cite}
\usepackage{amsmath}
\usepackage{amsfonts}
\usepackage{mathtools}
\usepackage{amssymb}
\usepackage{amsthm}
\usepackage[ruled,linesnumbered,noend]{algorithm2e}
\usepackage{cases}
\usepackage{graphicx}
\usepackage{bbm}
\usepackage{subcaption}

\newtheorem{theorem}{{\bf Theorem}}

\newtheorem{lemma}{{\bf Lemma}}

\newtheorem{corollary}{{\bf Corollary}}

\newtheorem{defn}{{\bf Definition}}

\IEEEoverridecommandlockouts

\begin{document}
%
\title{Optimal Slicing and Scheduling with Service Guarantees in Multi-Hop Wireless Networks}
\author{Nicholas Jones and Eytan Modiano\\
Laboratory for Information and Decision Systems, MIT\\
jonesn@mit.edu,  modiano@mit.edu
\thanks{This material is based upon work supported by the Department of the Air Force under Air Force Contract No. FA8702-15-D-0001. Any opinions, findings, conclusions or recommendations expressed in this material are those of the author(s) and do not necessarily reflect the views of the Department of the Air Force.}}

\IEEEaftertitletext{\vspace{-0.6\baselineskip}}
\maketitle

\begin{abstract}
    We analyze the problem of scheduling in wireless networks to meet end-to-end service guarantees. Using network slicing to decouple the queueing dynamics between flows, we show that the network’s ability to meet hard throughput and deadline requirements is largely influenced by the scheduling policy. We characterize the feasible throughput/deadline region for a flow under a fixed route and set of slices, and find throughput- and deadline-optimal policies for a solitary flow. We formulate the feasibility problem for multiple flows in a general topology, and show its equivalence to finding a bounded-cost cycle on an exponentially large graph, which is unsolvable in polynomial time by the best-known algorithm. Using a novel concept called delay deficit, we develop a sufficient condition for meeting deadlines as a function of inter-scheduling times, and show that regular schedules are optimal for satisfying this condition. Motivated by this, we design a polynomial-time algorithm that returns an (almost) regular schedule, optimized to meet service guarantees for all flows.
\end{abstract}

\maketitle

\section{Introduction}~\label{sec:intro}

Future wireless networks will need to make stringent throughput and delay guarantees to support emerging technologies, including real-time control and virtual reality systems. These guarantees take the form of service agreements, and due to their critical nature, network providers must move beyond best-effort service to ensure the agreements are met. As a result, work has begun on new methods of handling such traffic using network slicing~\cite{tmobile2023}. The 5G standard contains support for Quality of Service (QoS) guarantees at the flow level~\cite{5ghub}, including guarantees on maximum latency seen by any packet up to a given throughput. A largely open question, however, is how to make these service guarantees in wireless networks with limited resources, unreliable links, and interference constraints. 

In this work we follow a similar approach, using network slicing to decouple the queueing dynamics between traffic flows. This is not only practical for making service guarantees at the flow level, but useful for highlighting how wireless interference affects a network's ability to meet these guarantees. By understanding the impact of interference, we can enable networks to meet service requirements with fewer resources and to support more flows. To illustrate how interference affects delay, consider a line network with a flow traveling from node $1$ to node $4$ as in Figure~\ref{fig:scheduling-delay}, and for ease of exposition assume an interference model where only one link can transmit at a time.

\begin{figure}
    \centering
    \includegraphics[width=0.45\textwidth]{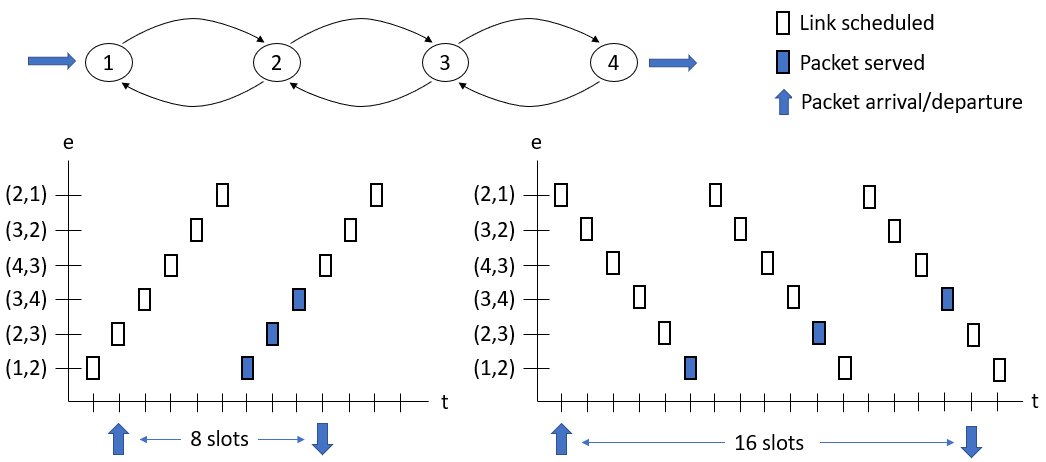}
    \caption{Effect of Scheduling Order on Packet Delay}
    \label{fig:scheduling-delay}
\end{figure}

The bottom of the figure shows two round-robin scheduling policies, and the worst-case end-to-end delay that a packet sees in each case, ignoring any queueing delay. The packet sees vastly different delays between the two policies, and the reason is clear. On the left it is served in three consecutive slots, while on the right it must wait five slots each time it is served before being served again. This simple example shows the impact of \textit{scheduling order} on packet delays. The link activation rates and the packet's route remain the same in both examples, but ordering the schedule in the direction of the flow's route has a significant impact on worst-case delay. As the size of the network grows, more complex interference models are considered, and more flows are introduced with different routes, this impact becomes even more pronounced, while at the same time harder to analyze. Motivated by this, we study the impact of interference and scheduling order on packet delays. Then, coupled with network slicing, we seek to characterize feasible service guarantees and to find supporting policies in general wireless networks.

There is a large body of work on wireless scheduling for maximizing throughput~\cite{tassiulas1990stability,jain2003impact,hajek_link_1988,kodialam2003characterizing}. In particular, Hajek and Sasaki~\cite{hajek_link_1988} designed a novel algorithm to find a minimum schedule length which meets a set of link demands in polynomial time. Kodialam et al.~\cite{kodialam2003characterizing} used Shannon's algorithm for coloring a multigraph to efficiently find schedules that are guaranteed to achieve at least $2/3$ of the maximum throughput, again in polynomial time.

There has also been considerable work done on making QoS guarantees in networks. One of the first approaches was Cruz's network calculus~\cite{cruz1991calculus,cruz1991calculus2}, which characterizes the average rate and burstiness of packet arrivals using a traffic shaping envelope, and then uses the convolution of service processes to bound the delay each packet experiences over multiple hops in a wired network. Several works have extended this framework to the wireless setting using a variant called \textit{stochastic} network calculus~\cite{fidler2006end, burchard2006min, li2007network, al2013min}, which bounds the tails of arrival and service processes to obtain a high-probability bound on end-to-end delay.

A novel QoS framework for single-hop wireless networks in the stochastic setting was developed by Hou and Kumar~\cite{hou2009theory}. It assumes a constant arrival process from each source and a strict deadline by which each packet must be delivered over an unreliable channel. By tracking the rates at which packets from each source are delivered, they design a policy to ensure the ``delivery ratio,'' or time average fraction of packets which are delivered by their deadline, meets a reliability requirement. They extend this framework in~\cite{hou2010utility} to solve utility maximization, and in~\cite{hou2013scheduling} to support Markov arrival processes. Several works have extended a version of this framework to multi-hop. In~\cite{li2012scheduling}, the authors analyze a multi-hop network with end-to-end deadline constraints and develop policies to meet delivery ratio requirements over wired links. In~\cite{liu2019spatial}, the authors design a spatio-temporal architecture with virtual links to solve a similar problem. In~\cite{singh2018throughput} and~\cite{singh2021adaptive}, the authors analyze a multi-hop wireless network with unreliable links. By considering each packet individually, they develop both centralized and decentralized policies for maximizing throughput with hard deadlines, using relaxed link capacity constraints and assuming no interference.

The closest to our work is~\cite{djukic2007quality} and~\cite{djukic_delay_2009}, which consider a multi-hop wireless network with interference constraints and find transmission schedules to meet traffic deadlines under constant arrivals. They show that when the direction of traffic is uniform, the optimal ordering can be found in polynomial time, and they develop heuristics for when the traffic direction is not uniform. The efficiency of these schedules was improved in~\cite{chilukuri_delay-aware_2015} by optimizing slot re-use for non-interfering links. The authors of~\cite{cappanera_link_2009} and~\cite{cappanera_efficient_2011} generalize the arrival processes to calculus-style envelopes and consider the case of sink-tree networks, while~\cite{cappanera_optimal_2013} incorporates routing.

Each paper in this line of work considers policies where links are scheduled in one contiguous block of time slots per scheduling period, which reduces the problem complexity but leads to end-to-end packet delays that grow with the schedule length. In this paper, we generalize and extend this line of work by designing scheduling policies without this restriction, deriving feasibility conditions on service guarantees, and incorporating network slicing. Using a novel concept called \textit{delay deficit}, we are able to meet much tighter deadlines without sacrificing throughput or adding complexity. Our main contributions can be summarized as follows.
\vspace{2pt}
\begin{itemize}
    \item In Section~\ref{sec:isolatedflows}, we develop necessary conditions on feasibililty for throughput and deadline guarantees, and both throughput- and deadline-optimal policies for a solitary flow under general interference.
    \item In Section~\ref{sec:interferencemodels}, we derive exact solutions to these optimal policies for three specific interference models. Focusing primarily on primary interference, we show that packet delays can grow linearly with the schedule length if the schedule order is not carefully designed.
    \item In Section~\ref{sec:generalflows}, we show that when deadlines are non-trivial, solving the feasibility problem for service guarantees in general networks is intractable, and we prove an alternative upper bound on packet delays based on schedule regularity.
    \item Finally, in Section~\ref{sec:linkscheduling}, we develop a polynomial-time algorithm to construct (almost) regular schedules, which are guaranteed to satisfy throughput and delay requirements without overprovisioning slices.
\end{itemize}

\vspace{5pt}
\section{Preliminaries}~\label{sec:sysmodel}
\subsection{System Model}

We consider a wireless network with fixed topology modeled as a directed graph $G=(V,E)$. Each link $e \in E$ has a fixed capacity $c_e$ and can transmit up to this number of packets in each time slot it is activated. Because links are wireless and share a wireless channel, they are subject to interference, which restricts the sets of links that can be scheduled at the same time. Time is slotted, with the duration of one slot equal to the transmission time over each link, so in each time slot the controller chooses a non-interfering set of links to be active. 

We consider a set of interference models $\Phi$, which take the form of $\phi$-hop interference constraints. Specifically, for any $\phi$, no two links can be activated at the same time if they are separated by fewer than $\phi$ hops. We will also refer to the interference model itself as $\phi \in \Phi$. The main focus of this work is primary interference, where $\phi=1$, but our framework can be generalized to any interference model in $\Phi$, and we present general results where possible. In the special cases of no interference and total interference, where only one link can be activated at a time, we define $\phi$ to be $0$ and $|E|-1$ respectively. Finally, denote $M^{\phi}$ as the set of feasible link activations, i.e., non-interfering links which can be activated simultaneously, under the model $\phi$.

Traffic arrives at the network in the form of flows, which can represent a single customer or an aggregation of customers. Denote flow $i$ as $f_i$ and the set of flows as $\mathcal{F}$. Arrivals are deterministic with $\lambda_i$ packets from $f_i$ arriving at the beginning of each slot\footnote{This can be generalized to a network calculus-style envelope, and for simplicity of exposition we assume traffic shaping occurs before packets arrive at the source.}. Each packet belonging to $f_i$ has a deadline $\tau_i$, and we say that a packet meets its deadline if it is delivered to its destination within $\tau_i$ slots of when it arrives, otherwise it expires. Flow $f_i$ is assigned a fixed pre-determined route $T^{(i)} \in \mathcal{T}$ from source to destination, where $\mathcal{T}$ is the set of all routes assigned to a flow. We denote $T^{(i)}_j$ as the $j$-th hop in the route $T^{(i)}$. The set of flows is fixed over a finite horizon $T$, which is independent of packet deadlines and assumed to be substantially larger. Assume that packet arrivals stop at time $T$, but packets remaining in the network are given time to be served by their respective deadlines. 

Each link $e \in T^{(i)}$ reserves capacity for $f_i$ in the form of a network slice, with capacity, or \textit{slice width}, equal to $w_{i,e}$. We will also refer to the slice itself as $w_{i,e}$ when there is no risk of confusion. Denote the vector of slice widths as $\boldsymbol{w} = \{ w_{i,e}, \ \forall f_i \in \mathcal{F}, e \in E \}$. Because link capacities are fixed, the sum of all slice widths allocated on link $e$ must be bounded by $c_e$. Each slice has its own first-come-first-served queue, which decouples the queueing dynamics between flows. Let $Q_{i,e}(t)$ be the size of the queue belonging to slice $w_{i,e}$ at the beginning of time slot $t$, after packets have arrived but before any packets are served. We assume the network is empty before $t=0$, so $Q_{i,e}(t) = 0$ for all $i$ and $e$, and $t < 0$.

\vspace{2pt}
\subsection{Policy Structure}\label{sec:policy}

Define $\Pi$ as the set of admissible scheduling policies, which are work-conserving and satisfy the interference constraints $\phi$. Let $\mu^{\pi}(t) \in M^{\phi}$ be the set of links activated at time $t$ under $\pi$, and let $\mu_e^{\pi}(t) = 1$ if $e \in \mu^{\pi}(t)$, and $0$ otherwise. The work-conserving property ensures that each link $e \in \mu^{\pi}(t)$ serves the smaller of its queue size and its slice capacity each time it is activated. We are interested in policies which meet the following criteria.

\begin{defn}
    A policy $\pi \in \Pi$ \textit{supports} a set of flows $\mathcal{F}$ on a time interval if and only if no packets expire during that interval. If no interval is specified, it is assumed to hold for all $t \geq 0$.
\end{defn}

Specifically, we would like to understand how wireless interference affects scheduling delay, what sets of flows can be supported in general, and how to design policies which support a set of flows while minimizing slice widths, thus leaving as much capacity as possible for best-effort traffic.

Without wireless interference, this problem would be trivial because there is no scheduling component. Under our model of constant arrivals with fixed routes and slicing, each link can forward a slice width of packets at every time step. Then the network need only guarantee that slice widths satisfy both throughput and link capacity constraints, and that the number of hops in each flow's route is less than its deadline. In the wireless setting, this problem becomes nontrivial and interesting. The scheduling policy, which is a function of the interference, dictates how often and in what order links are scheduled. The less frequently a link is served, the larger its slice width must be for a given throughput, and the more delay a packet sees at that link in the worst case. 

Packets arrive at a fixed rate in our model, so the dynamics of the system can be described exactly under a fixed policy in $\Pi$. A consequence of this fact combined with finite deadlines is that queue sizes are bounded under any policy that supports $\mathcal{F}$. This allows us to analyze the inherently difficult problem of delay guarantees in multi-hop wireless networks in a more tractable way. In particular, it leads to the following result.

\begin{theorem}\label{th:cyclicschedules}
    If there exists a policy in $\Pi$ that supports $\mathcal{F}$, then there must exist at least one such policy $\pi$ that is cyclic with period $K^{\pi}$, so that 
    \begin{equation}
        \mu^{\pi}(t) = \mu^{\pi}(t+K^{\pi}), \ \forall \ t \geq 0,
    \end{equation}
    and under $\pi$,
    \begin{equation}
        Q_{i,e}^{\pi}(t) = Q_{i,e}^{\pi}(t+K^{\pi}), \ \forall \ f_i \in \mathcal{F}, \ e \in E, \ t_0 \leq t \leq T,
    \end{equation}
    for some $t_0 > 0$ and sufficiently large $T$.
\end{theorem}

\begin{proof}
    See Appendix~\ref{app:cyclicschproof}.
\end{proof}

This result also provides the following necessary and sufficient condition for meeting deadlines.

\begin{corollary}\label{cor:deadlinecondition}
    A cyclic policy $\pi \in \Pi$ with sufficiently large $T$ supports $\mathcal{F}$ if and only if
    \begin{equation}\label{eq:queuebound}
        \max_{t \geq 0} \sum_{e \in T^{(i)}} Q_{i,e}^{\pi}(t) \leq \lambda_i \tau_i, \ \forall \ f_i \in \mathcal{F}.
    \end{equation}
\end{corollary}

\begin{proof}
    Because $T$ and $\lambda_i$ are finite, the sum is bounded and the maximum is guaranteed to exist. From the proof of Theorem~\ref{th:cyclicschedules}, the condition~\eqref{eq:queuebound} is necessary and sufficient for any $\pi$ to support $\mathcal{F}$ on the interval $0 \leq t \leq T$ because of the constant arrivals and FCFS queues. The proof also shows that any policy which supports $\mathcal{F}$ on the interval $t \leq T$ also supports $\mathcal{F}$ for $t > T$ by assuming dummy arrivals until all real packets have been delivered. Provided $T$ is sufficiently large for $\pi$ to have completed at least one cycle, the queue sizes with dummy arrivals will be no larger than those before time $T$, and the result follows.
\end{proof}

Corollary~\ref{cor:deadlinecondition} is useful for finding and verifying that a policy supports a set of flows. Once a policy is fixed, the system is completely deterministic, and iteratively solving the queue evolution equations allows one to track queue sizes at each $t$, verifying that deadlines are being met without keeping track of the delay of each individual packet. The utility of this result as well as Theorem~\ref{th:cyclicschedules} motivate us to consider only cyclic policies in the remainder of this work, without loss of optimality. Denote this class of policies as $\Pi_c \subseteq \Pi$. We will also assume in the remainder of this work that $T$ is sufficiently large for any policy to have completed at least one cycle.

Under a policy $\pi \in \Pi_c$, define the time average activation frequency of link $e$ as 
\begin{equation}
    \bar{\mu}_e^{\pi} \triangleq \frac{1}{K^{\pi}} \sum_{t=0}^{K^{\pi}} \mu_e^{\pi}(t),
\end{equation}
and the number of activations per scheduling period as $\eta_e^{\pi} \triangleq \bar{\mu}_e^{\pi} K^{\pi}$. Finally, let the time average service rate of slice $w_{i,e}$ be
\begin{equation}
    \bar{w}_{i,e}^{\pi} \triangleq \bar{\mu}_e^{\pi} w_{i,e}.
\end{equation}
We will make heavy use of these quantities in the analysis going forward.

\section{Feasibility}\label{sec:isolatedflows}

Having defined our policy class, we now turn to characterizing the feasibility region. Define the feasible region for a route $T^{(i)}$ with slice widths $\boldsymbol{w}$ as the set of all throughput/deadline pairs which a flow can achieve under $\phi$ and any scheduling policy in $\Pi_c$, and denote this region as $\Lambda^{\phi}_i(\boldsymbol{w})$. We define this region for a given route to show the feasible throughput/deadline guarantees that can be made to a flow on that route by optimizing the scheduling policy independently of other flows. In the general setting, with many flows on separate routes coupled through the scheduling policy, the jointly achievable region is defined as 
\begin{equation}\label{eq:feasregion}
    \Lambda^{\phi} (\boldsymbol{w}) \subseteq \prod_{T^{(i)} \in \mathcal{T}} \Lambda^{\phi}_i (\boldsymbol{w}),
\end{equation}
where the scheduling policy must optimize over all flows jointly. Analyzing $\Lambda^{\phi}_i (\boldsymbol{w})$ allows us to define optimal policies for a single flow, as well as bounds on feasibility in the general case. We begin by characterizing throughput optimality.

\subsection{Throughput Optimality}

Define $\lambda_i^*(\pi,\boldsymbol{w})$ as the maximum throughput a policy $\pi$ can support on $T^{(i)}$ with slice widths $\boldsymbol{w}$, given by the following intuitive result. 

\begin{lemma}\label{lemma:lambdabound}
    For any admissible policy $\pi \in \Pi_c$,
    \begin{equation}\label{eq:lambdabound}
        \lambda_i^*(\pi, \boldsymbol{w}) = \min_{e \in T^{(i)}} \bar{w}_{i,e}^{\pi}.
    \end{equation}
\end{lemma}

\begin{proof}
    Assume for contradiction that $\lambda_i > \bar{w}_{i,e}^{\pi}$ for some $f_i$ and $e$. In one scheduling period $K^{\pi}$ for $t \leq T$, $\lambda_i K^{\pi}$ packets are added to the queue and at most $\bar{w}_{i,e}^{\pi} K^{\pi}$ packets are served. Then $Q_{i,e}(t+K^{\pi}) > Q_{i,e}(t)$ for all $t \leq T$, so from Theorem~\ref{th:cyclicschedules}, either $\pi \notin \Pi_c$ or does not support $f_i$, which is a contradiction.
\end{proof}

Now define a throughput-optimal policy as follows.

\begin{defn}
    A policy $\pi \in \Pi_c$ is throughput-optimal for a route $T^{(i)}$ and slice widths $\boldsymbol{w}$ if $\lambda_i^*(\pi,\boldsymbol{w}) \geq \lambda_i^*(\pi',\boldsymbol{w})$ for all $\pi' \in \Pi_c$.
\end{defn}

Because $\lambda_i^*(\pi,\boldsymbol{w})$ is the largest supported throughput for a given $\pi$, a throughput-optimal policy is one that maximizes this quantity over all $\pi \in \Pi_c$. In particular, it is any solution to 
\begin{align}\label{eq:thoptformulation}
\begin{aligned}
    \max_{\pi \in \Pi_c} &\min_{e \in T^{(i)}} \bar{\mu}_e^{\pi} w_{i,e} \\
    \text{s.t.} \ &\mu^{\pi}(t) \in M^{\phi} , \ \forall \ 0 \leq t \leq K^{\pi},
\end{aligned}
\end{align}
and we denote the solution as $\lambda_i^*(\boldsymbol{w})$. We will solve this problem exactly for different interference models in the next section, but note that on average the optimization tries to drive $\bar{\mu}_e^{\pi} w_{i,e}$ to equality along all links, so in general links wth smaller slices tend to be activated more frequently.

\subsection{Deadline Optimality}

Define $\tau_i^*(\pi, \boldsymbol{w}, \lambda_i)$ as the smallest deadline which the network can guarantee to a flow on $T^{(i)}$ with slice widths $\boldsymbol{w}$ and a throughput of $\lambda_i$ under policy $\pi$. Equivalently, we refer to $\tau_i^*$ as the largest end-to-end delay experienced by any packet in the flow. In a similar way, define $\tau_i^*(\pi) \triangleq \lim_{\lambda_i \to 0} \tau_i^*(\pi, \boldsymbol{w}, \lambda_i)$ as the smallest deadline that $\pi$ can guarantee for any $\lambda_i > 0$. Note that because $\lambda_i$ can be arbitarily close to zero, $\tau_i^*(\pi)$ is independent of $\boldsymbol{w}$. Because of the equivalence of maximum packet delays and minimum guaranteed deadlines, we will refer to $\tau_i^*$ interchangeably by either of these definitions.

To derive a lower bound on $\tau_i^*(\pi)$, we begin by introducing the concept of inter-scheduling times. Denote the set of time slots where link $e$ is scheduled under a policy $\pi$ as $\mathcal{T}_e^{\pi} \triangleq \{ t \geq 0 \ | \ \mu_e^{\pi}(t) = 1 \}$. Then define the minimum inter-scheduling time $\underline{k}_{e,e+1}^{\pi}$ of links $e$ and $e+1$ to be the smallest time interval between consecutive scheduling events of links $e$ and $e+1$, in that order, so that
\begin{equation}
    \underline{k}_{e,e+1}^{\pi} \triangleq \min_{t_e \in \mathcal{T}_e^{\pi}} \min_{t_{e+1} > t_e : t_{e+1} \in \mathcal{T}_{e+1}^{\pi}} (t_{e+1} - t_e),
\end{equation}
and the maximum inter-scheduling time $\overline{k}_{e,e+1}^{\pi}$ to be the largest such time, defined as
\begin{equation}
    \overline{k}_{e,e+1}^{\pi} \triangleq \max_{t_e \in \mathcal{T}_e^{\pi}} \min_{t_{e+1} > t_e : t_{e+1} \in \mathcal{T}_{e+1}^{\pi}} (t_{e+1} - t_e).
\end{equation}

We can also speak of the inter-scheduling times of a single link $e$ as the times between consecutive scheduling events of that link, and denote this as $\underline{k}_e^{\pi}$ and $\overline{k}_e^{\pi}$ respectively for ease of notation. Then we can lower bound minimum deadlines as follows.

\begin{lemma}\label{lemma:taulowerbound}
    For any admissible policy $\pi \in \Pi_c$,
    \begin{equation}\label{eq:taulowerbound}
        \tau_i^*(\pi) \geq \underline{k}_0^{\pi} + \sum_{0 \leq j < |T^{(i)}|-1} \underline{k}_{j,j+1}^{\pi} + 1,
    \end{equation}
    where link $j = T^{(i)}_j$ for all $j$.
\end{lemma}

\begin{proof}
    Any packet which is delivered to its destination must traverse all links on its route in order, and under a policy $\pi$, every packet served by link $e$ that arrives at link $e+1$ must wait at least $\underline{k}_{e,e+1}^{\pi}$ slots before being served. Therefore, the smallest amount of time between being served at the source link and being served at the destination link is $\sum_{0 \leq e < |T^{(i)}|-1} \underline{k}_{e,e+1}^{\pi}$. At least $\lambda_i$ packets must wait $\underline{k}_0^{\pi}$ slots from when they arrive at the source link until they are served, and it takes one slot to be delivered once served at the destination link. The result follows.
\end{proof}

This bound begins to formalize the idea that schedule order plays an important role in minimizing delay. In particular, it motivates us to minimize inter-scheduling times between consecutive links on a route to keep packet delays small. We will show that this results in a deadline-optimal policy, per the following definitions.
\begin{defn}
    A policy $\pi \in \Pi_c$ is deadline-minimizing for a route $T^{(i)}$, slice widths $\boldsymbol{w}$, and throughput $\lambda_i$ if $\tau_i^*(\pi, \boldsymbol{w}, \lambda_i) \leq \tau_i^*(\pi', \boldsymbol{w}, \lambda_i)$ for all $\pi' \in \Pi_c$.

    Furthermore, a policy $\pi \in \Pi_c$ is deadline-optimal for a route $T^{(i)}$ if $\tau_i^*(\pi) \leq \tau_i^*(\pi')$ for all $\pi' \in \Pi_c$.
\end{defn}

To avoid confusion, we distinguish between the terms \textit{deadline-minimizing} when speaking in terms of a specific throughput, and \textit{deadline-optimal} when speaking independently of throughput. 

We will show that deadline optimality is achieved by a subclass of $\Pi_c$ we call \textit{ordered round-robin} (ORR) scheduling policies, which minimize the bound in~\eqref{eq:taulowerbound} while showing that it is tight. Denote the ORR policy for $T^{(i)}$ under an interference model $\phi$ as $ORR^{\phi}(i)$. Recalling that links separated by fewer than $\phi$ hops cannot be scheduled simultaneously, define the ORR policy as follows. At each time $t$, activate link $T^{(i)}_j$, where $j=t \bmod \phi$, along with every $\phi+1$ subsequent links. Then at each slot, links at equally spaced intervals of $\phi+1$ hops are activated, beginning with hops $\{0, \phi+1, 2 \phi+2,\dots \}$ at time $t=0$, hops $\{1, \phi+2, 2 \phi+3, \dots \}$ at $t=1$, and so on. Note that this schedule has a period of $\phi+1$.

\begin{theorem}\label{th:orrbound}
    The ORR policy is deadline-optimal under any interference model $\phi \in \Phi$, with a maximum packet delay
    \begin{equation}
        \tau_i^* \big(ORR^{\phi}(i) \big) = |T^{(i)}|+\phi,
    \end{equation}
    and activation rates
    \begin{equation}
        \bar{\mu}_e^{ORR^{\phi}(i)} = \frac{1}{\phi+1}, \ \forall \ e \in T^{(i)}.
    \end{equation}
\end{theorem}

\begin{proof}
    See Appendix~\ref{app:orrproof}.
\end{proof}

Note that the ORR policy meets the bound in~\eqref{eq:taulowerbound} with equality, because the minimum inter-scheduling times are $\underline{k}_0 = \phi$ and $\underline{k}_{e,e+1} = 1$ for all $e$. From~\eqref{eq:lambdabound}, the maximum throughput under the ORR policy is
\begin{equation}\label{eq:orrlambda}
    \lambda_i^* \big(ORR^{\phi}(i), \boldsymbol{w} \big) = \frac{1}{\phi+1} \min_{e \in T^{(i)}} w_{i,e},
\end{equation}
which is not throughput-optimal in general, and highlights the tradeoff between maximizing throughput and minimizing delay. When slice widths are equal, however, this tradeoff does not occur.

\begin{corollary}\label{cor:jointoptimalpoint}
    When slice widths are equal across $T^{(i)}$, the $ORR^{\phi}(i)$ policy is both throughput-optimal and deadline-optimal.
\end{corollary}

\begin{proof}
    See Appendix~\ref{app:jointoptproof}.
\end{proof}

\section{Interference Models}\label{sec:interferencemodels}

Having characterized bounds on the feasibility region and the structure of throughput- and deadline-optimal policies under general interference, we now examine three specific interference models and derive exact results for these cases.

\subsection{No Interference}

When there is no interference, the system is equivalent to a wired network. Every link can be activated in each slot, so $\bar{\mu}_e = 1$ for all $e$, and $\lambda_i^*(\boldsymbol{w}) = \min_{e \in T^{(i)}} w_{i,e}$ from Lemma~\ref{lemma:lambdabound}. Observe that this is equivalent to an $ORR^0(i)$ policy, so $\tau_i^* = |T^{(i)}|$ from Theorem~\ref{th:orrbound}. Then the feasible region of flow $i$ can be defined as
\begin{equation}
    \Lambda_i^0(\boldsymbol{w}) \triangleq \{ (\tau_i,\lambda_i) \ | \ \tau_i \geq |T^{(i)}|, 0 \leq \lambda_i \leq \min_{e \in T^{(i)}} w_{i,e} \},
\end{equation}
and because there is no interference,~\eqref{eq:feasregion} holds with equality and the overall feasible region $\Lambda^0(\boldsymbol{w}) = \prod_{T^{(i)} \in \mathcal{T}} \Lambda^0_i(\boldsymbol{w})$.

\subsection{Total Interference}

Next we consider the total interference scenario, where only one link can be active at a time. For flow $i$, this is equivalent to a $(|T^{(i)}|-1)$-hop interference model. Then from Theorem~\ref{th:orrbound}, $\tau_i^* = |T^{(i)}| + \phi = 2|T^{(i)}| - 1$, and this is achieved using an $ORR^{(T^{(i)}-1)}(i)$ policy. From~\eqref{eq:orrlambda},
\begin{equation}
    \lambda_i^*\big(ORR^{(T^{(i)}-1)}(i), \boldsymbol{w} \big) = \frac{1}{|T^{(i)}|} \min_{e \in T^{(i)}} w_{i,e}.
\end{equation}

The main focus of this work is primary interference, so we only briefly list the properties of throughput optimality for the total interference case, and we omit proofs due to space constraints. Similar to what we will see in the primary interference case, a throughput-optimal policy $\pi$ has a maximum throughput of 
\begin{equation}\label{eq:totalintflambda}
    \lambda_i^*(\pi, \boldsymbol{w}) = \frac{1}{|T^{(i)}|} H(\boldsymbol{w}),
\end{equation}
where $H(\cdot)$ represents the harmonic mean.

The achievable deadlines under a throughput-optimal policy are hard to characterize, but grow with $O(|T^{(i)}|^2)$. One can see this by examining the deadline condition from Corollary~\ref{cor:deadlinecondition}, and dividing both sides by $\lambda_i$. Plugging in~\eqref{eq:totalintflambda} gives the necessary and sufficient condition
\begin{equation}
    \Big\lceil \frac{|T^{(i)}|}{H(\boldsymbol{w})} \sum_{e \in T^{(i)}} Q_{i,e}(t) \Big\rceil \leq \tau_i,
\end{equation}
for meeting deadlines, where the ceiling is taken because deadlines are necessarily integer-valued. It is easy to see that this bound is $O(|T^{(i)}|^2)$.

\subsection{Primary Interference}
We finally turn to primary interference, which is the main focus of this work, and begin by characterizing the deadline-optimal point $\tau_{\min}$ in the feasibility region $\Lambda_i^1(\boldsymbol{w})$. From Theorem~\ref{th:orrbound}, recalling that $\phi=1$ under primary interference, the minimum achievable deadline is $\tau_i^* = |T^{(i)}| + \phi = |T^{(i)}|+1$, which is again achieved using an $ORR^1(i)$ policy. Then from~\eqref{eq:orrlambda}, we have
$\lambda_i^* \big(ORR^1(i), \boldsymbol{w} \big) = \frac{1}{2} \min_{e \in T^{(i)}} w_{i,e}$, and
\begin{equation}
    \tau_{\min} \big( \Lambda_i^1(\boldsymbol{w}) \big) \triangleq \Big( |T^{(i)}|+1, \frac{1}{2} \min_{e \in T^{(i)}} w_{i,e} \Big).
\end{equation} 

Obtaining the throughput-optimal point $\lambda_{\max}$ is less trivial. A throughput-optimal policy can be derived from~\eqref{eq:thoptformulation} by relaxing the link activation constraint to be $\bar{\mu}_e + \bar{\mu}_{e+1} \leq 1$ for all adjacent link pairs $(e,e+1)$. In our single route scenario under primary interference, only adjacent links interfere with one another, so this is a necessary and sufficient condition for feasibility. Then a throughput-optimal set of activations is any solution to
\begin{align}
\begin{aligned}
    \max_{\pi \in \Pi_c} &\min_{e \in T^{(i)}} \bar{\mu}_e^{\pi} w_{i,e} \\
    \text{s.t.} \ &\bar{\mu}_e^{\pi} + \bar{\mu}_{e+1}^{\pi} \leq 1, \ \forall \ 0 \leq e < |T^{(i)}|-1,
\end{aligned}
\end{align}
which can be solved exactly.

\begin{theorem}
    Any throughput-optimal policy $\pi \in \Pi_c$ on $T^{(i)}$ under primary interference constraints has a maximum throughput of 
    \begin{equation}
        \lambda_i^*(\pi, \boldsymbol{w}) = \min_{e \in T^{(i)}} \frac{w_{i,e}{w_{i,e+1}}}{w_{i,e} + w_{i,e+1}}
    \end{equation}
\end{theorem} 

\begin{proof}
    As noted previously, the interference constraint for a single route under primary interference is $\bar{\mu}_e + \bar{\mu}_{e+1} \leq 1$ for all $(e,e+1)$, and from Lemma~\ref{lemma:lambdabound}, $\lambda_i^*(\pi,\boldsymbol{w}) = \min_{e \in T^{(i)}} \bar{w}_{i,e}^{\pi}$. Then the maximum throughput which can pass through $(e,e+1)$ is the solution to
    \begin{align}
    \begin{aligned}\label{eq:linkpairthopt}
        \max_{\bar{\mu}_e,\bar{\mu}_{e+1}} \ &\min \{\bar{\mu}_e w_{i,e}, \bar{\mu}_{e+1} w_{i,e+1} \} \\
        \text{s.t.} \ &\bar{\mu}_e + \bar{\mu}_{e+1} \leq 1.
    \end{aligned}
    \end{align}
    Because $\bar{\mu_e}$ and $\bar{\mu_{e+1}}$ are continuous, the solution occurs when $\bar{\mu}_e + \bar{\mu}_{e+1} = 1$ and $\bar{\mu}_e w_{i,e} = \bar{\mu}_{e+1} w_{i,e+1}$. Solving this set of equations yields
    \begin{equation}\label{eq:linkpairthoptsol}
        \bar{\mu}_e = \frac{w_{i,e+1}}{w_{i,e}+w_{i,e+1}}, \ \bar{\mu}_{e+1} = \frac{w_{i,e}}{w_{i,e}+w_{i,e+1}},
    \end{equation}
    and the maximum throughput which can pass through the pair of links is
    \begin{equation}
        \bar{\mu}_e w_{i,e} = \bar{\mu}_{e+1} w_{i,e+1} = \frac{w_{i,e} w_{i,e+1}}{w_{i,e}+w_{i,e+1}}.
    \end{equation}

    Now define the pair of links $(e^*,e^*+1)$ which minimize this quantity as the bottleneck link pair, and $\lambda_i^*(\pi)$ as the bottleneck throughput. Each link in $T^{(i)}$ belongs to two pairs, except the first and last link on the route. Define a set of throughput-optimal activations as follows. For each link $e$, find the two throughput-optimal activation rates from~\eqref{eq:linkpairthoptsol} belonging to its two link pairs, and assign the smaller of these to $e$, so that 
    \begin{multline}\label{eq:minneighboractivation}
        \bar{\mu}_e^{\pi^*} = \min \Big\{ \frac{w_{i,e-1}}{w_{i,e}+w_{i,e-1}}, \ \frac{w_{i,e+1}}{w_{i,e}+w_{i,e+1}} \Big\}, \ \forall 1 \leq e < -1, \\ \bar{\mu}_0^{\pi^*} = \frac{w_{i,1}}{w_{i,0}+w_{i,1}}, \ \bar{\mu}_{-1}^{\pi^*} = \frac{w_{i,-2}}{w_{i,-1}+w_{i,-2}},
    \end{multline}
    where we slightly abuse notation to denote $|T^{(i)}|-j$ as $-j$. Then each link supports at least as much throughput as $\lambda_i^*(\pi,\boldsymbol{w})$, because by definition this is the bottleneck throughput, and this set of activations is guaranteed to meet scheduling constraints because each activation pair in~\eqref{eq:linkpairthoptsol} is a feasible solution to~\eqref{eq:linkpairthopt}, and each link in the pair either uses this rate or a strictly smaller one. This completes the proof.
\end{proof}

If the activation rates are feasible, algorithms exist to find a throughput-optimal policy $\pi^*$ in polynomial time~\cite{hajek_link_1988}. In order to ensure that rates are met and the number of activations $\eta_e^{\pi^*} = \bar{\mu}_e^{\pi^*} K^{\pi^*}$ is integer, the schedule length $K^{\pi^*}$ must be at least as large as the least common multiple of the denominators of the activation rates, which can become arbitrarily large.

While $\pi^*$ is guaranteed to be throughput-optimal if it exists, it is not deadline-minimizing in general. Finding a deadline-minimizing policy for this throughput is more difficult because, while throughput optimality requires only time average constraints, we have seen that deadline guarantees are largely dependent on schedule order. When slice widths are equal, $\bar{\mu}_e^{\pi^*} = \frac{1}{2}$ for all $e$, and the throughput-optimal point converges to the deadline-optimal (and therefore deadline-minimizing) point as described in Corollary~\ref{cor:jointoptimalpoint}. For general slice widths, we have the following result.

\begin{theorem}\label{th:singleflowcomp}
    Solving for a deadline-minimizing policy $\pi$ and corresponding $\tau_i^*(\pi, \boldsymbol{w}, \lambda_i)$, for a given throughput $\lambda_i$, has complexity which is exponential in $|T^{(i)}|$.
\end{theorem}

\begin{proof}
    In the general case, finding a deadline-minimizing policy is equivalent to finding a min-max cost cycle on a graph. Define a directed graph $G^{(i)} = (V^{(i)}, E^{(i)})$, where each vertex $v \in V^{(i)}$ represents a possible state of queue sizes across the route $T^{(i)}$, and each edge $e \in E^{(i)}$ represents a valid action, or set of links which can be activated together. An edge $e$ exists from state $s_t$ to $s_{t+1}$ if taking action $e$ while in state $s_t$ leads to state $s_{t+1}$.

    Define the cost of each edge $e$ that ends in state $s_t$ to be $\xi_e = \sum_{l \in T^{(i)}} Q_{i,l}(t)$, where $s_t$ corresponds to the set of queue sizes $\{ Q_{i,l}(t), \ \forall l \in T^{(i)} \}$. Because our policy is guaranteed to be cyclic, we require that $Q_{i,l}(t) = Q_{i,l}(t+K)$ for all $l \in T^{(i)}$ and a schedule of length $K$, and for all $t$ while the network is in steady state. Therefore, a valid policy in steady state is a cycle $\mathcal{C}^{(i)}$ on the graph $G^{(i)}$, and from Corollary~\ref{cor:deadlinecondition}, any such policy meets a deadline $\tau_i$ when $\max_{e \in \mathcal{C}^{(i)}} \xi_e \leq \lambda_i \tau_i$.

    Under this equivalent problem,
    \begin{equation}\label{eq:cycleopt}
        \tau_i^*(\pi, \boldsymbol{w}, \lambda_i) = \frac{1}{\lambda_i} \min_{\mathcal{C}^{(i)}} \max_{e \in \mathcal{C}^{(i)}} \xi_e
    \end{equation}
    for a throughput-optimal policy $\pi$. Solving~\eqref{eq:cycleopt} is at least as hard as deciding whether there exists a cycle $\mathcal{C}^*$ such that $\frac{1}{\lambda_i} \max_{e \in \mathcal{C}^*} \xi_e \leq \tau_i'$ for some $\tau_i'$, because the solution to~\eqref{eq:cycleopt} immediately provides a solution to the decision problem. This decision problem, in turn, is at least as hard as deciding whether there exists a $\mathcal{C}^*$ such that $\frac{1}{\lambda_i |\mathcal{C}^*|} \sum_{e \in \mathcal{C}^*} \xi_e \leq \tau_i'$, because if a cycle exists that satisfies the max condition it immediately satisfies the average condition. Combining these we observe that finding $\tau_i^*(\pi, \boldsymbol{w}, \lambda_i)$ is at least as hard as finding a cycle on $G^{(i)}$ with bounded average cost.

    In 1978, Karp developed an $O(|V^{(i)}||E^{(i)}|)$ algorithm for this problem~\cite{karp1978characterization}, which is still the most efficient known algorithm. Let $q_{e}$ be the number of discrete sizes that $Q_{i,e}$ can take, and let $q = \min_e q_e$ be the smallest such value. Then $|V^{(i)}| \geq q^{|T^{(i)}|}$. There is at least one valid action in each state, so we conclude that the best known algorithm for finding $\tau_i^*(\pi,\boldsymbol{w},\lambda_i)$ has complexity at least $O(q^{2|T^{(i)}|})$.
\end{proof}

The complexity of this problem illustrates the difficulty of minimizing deadlines for general policies even in this simplest case of a single route. This is especially unfortunate because, when schedule order is not optimized, packet delays can become large for a given schedule length and set of activation rates. To see this, we first introduce the concept of resource-minimizing slices, which follows from Lemma~\ref{lemma:lambdabound}. 

\begin{defn}
    A set of slices $\boldsymbol{w}$ is resource-minimizing for a given policy $\pi \in \Pi_c$ and a set of flows $\mathcal{F}$ if $w_{i,e} = \frac{\lambda_i}{\bar{\mu}_e^{\pi}}$ for all $f_i \in \mathcal{F}$ and $e \in T^{(i)}$.
\end{defn}

Note that this implies a resource-minimizing slice is the smallest feasible slice which can support the given throughput under $\pi$. This definition is useful for the following Lemma.

\begin{lemma}\label{lemma:emptyqueues}
    If a policy $\pi \in \Pi_c$ supports a set of flows $\mathcal{F}$ under primary interference, and slices are resource-minimizing, then for every $f_i \in \mathcal{F}$ and $e \in T^{(i)}$, there exists at least one slot $t_{i,e}$ in each scheduling period where $Q_{i,e}(t_{i,e}) = 0$.
\end{lemma}

\begin{proof}
    See Appendix~\ref{app:emptyqueueproof}.
\end{proof}

This result helps bound worst-case delay because slices can always be made larger without impacting performance, simply by treating the slice as if it was the original size. Therefore, without loss of generality we can restrict slices to be their smallest feasible size, i.e. resource-minimizing, in our worst-case delay analysis. Then Lemma~\ref{lemma:emptyqueues} says each queue must be emptied at least once per scheduling period under any policy which maximizes delay. This allows us to bound packet delay as follows.

\begin{theorem}\label{th:deadlinebound}
    Under any policy $\pi \in \Pi_c$,
    \begin{equation}
        \tau_i^*(\pi, \boldsymbol{w}, \lambda_i) \leq K^{\pi} \sum_{e \in T^{(i)}} (1 - \bar{\mu}_e^{\pi}) + 1
    \end{equation}
    under primary interference constraints, for any slice widths $\boldsymbol{w}$ and feasible throughput $0 < \lambda_i \leq \lambda_i^*(\pi,\boldsymbol{w})$. Moreover, there exists at least one policy where this bound is tight, and many policies where $\tau_i^*$ grows linearly with $K^{\pi}$.
\end{theorem}

\begin{proof}
    See Appendix~\ref{app:deadlineboundproof}.
\end{proof}

We can now formally define the throughput-optimal point (achievable under any schedule order) as 
\begin{multline}
    \lambda_{\max} \big( \Lambda_i^1(\boldsymbol{w}) \big) \triangleq \\
    \Big( K^{\pi} \sum_{e \in T^{(i)}} (1 - \bar{\mu}_e^{\pi}) + 1, \min_{e \in T^{(i)}} \frac{w_{i,e}{w_{i,e+1}}}{w_{i,e} + w_{i,e+1}} \Big).
\end{multline}

We have seen that $K^{\pi}$ can grow arbitarily large for general activation rates, so the bound in Theorem~\ref{th:deadlinebound} does not provide strong deadline guarantees. In fact, it shows that there exists at least one policy $\pi \in \Pi_c$ where packets experience delays within a constant factor of $K^{\pi} |T^{(i)}|$, and many policies where packets experience delays that grow with $K^{\pi}$. These policies are not pathological examples, but rather any policy with inter-scheduling times that grow with the schedule length. This includes policies that schedule each link in one contiguous block per scheduling period (as in~\cite{djukic2007quality} and the following line of work), or simply random orderings of activations which happen to have large inter-scheduling times. We formalize this idea in the next section, and to avoid delays that grow with $K^{\pi}$, we introduce policies with bounded inter-scheduling times. We show that by restricting ourselves to this class of policies, we can make deadline guarantees that are independent of the schedule length without sacrificing thoughput.

\begin{figure}
    \centering
    \includegraphics[width=0.45\textwidth]{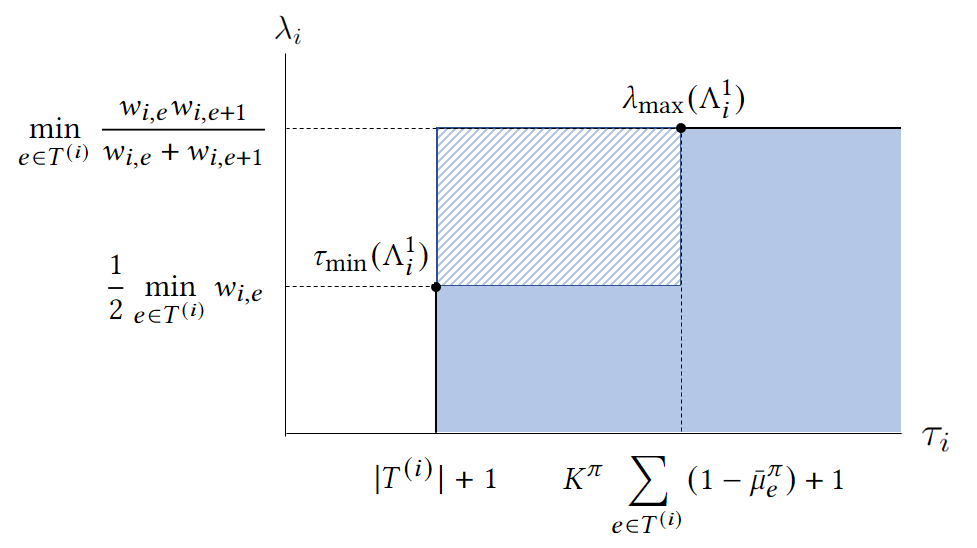}
    \caption{Feasible Region Under Primary Interference}
    \label{fig:lambda-primary-interf}
\end{figure}

In Figure~\ref{fig:lambda-primary-interf}, we show the feasible throughput/deadline region for a flow under primary interference, highlighting the deadline-optimal point and the throughput-optimal point with the upper bound on packet delay from Theorem~\ref{th:deadlinebound}. The boundary in the shaded region between these two operating points is difficult to characterize, as shown in Theorem~\ref{th:singleflowcomp}. However, we can show it is monotonic, and we provide a brief proof sketch.

Consider a network with some throughput $\lambda_i$, and with minimum deadlines and associated policy $\pi$. Now imagine a second identical network with some $\lambda_i' \leq \lambda_i$ under the same policy $\pi$. Each packet in the second network sees the same service opportunities as the first, and has the same number or fewer packets ahead of it in its queue at all times. Therefore, it cannot see a larger delay than what is seen by packets in the original network, which is by definition the original minimum deadline. The minimum deadline in the second network must be at least as small as the largest delay that any packet sees, so it must be at least as small as the original minimum deadline. Therefore, decreasing throughput cannot lead to larger deadlines, and the boundary of the region is monotonic between the throughput- and deadline-optimal points.

\section{Scheduling for Service Guarantees}\label{sec:generalflows}

We now move to the problem of scheduling multiple flows in a general network under primary interference, drawing on the feasibility results of the previous section. As a secondary goal, we seek to keep slices as close as possible to resource-minimizing. This allows the network to support more deadline-constrained traffic, as well as more best-effort traffic with the unallocated capacity. We have seen that it is challenging to design policies with deadline guarantees that are tighter than the universal upper bound in Theorem~\ref{th:deadlinebound}. This is true in the case of a solitary flow, as shown in Theorem~\ref{th:singleflowcomp}, and certainly remains true in the general case with multiple flows.

\begin{theorem}\label{th:generalcomplexity}
    Given a set of flows $\mathcal{F}$ with non-trivial deadlines (i.e., tighter than the bound in Theorem~\ref{th:deadlinebound}), finding a policy $\pi \in \Pi_c$ that supports $\mathcal{F}$ has complexity which is exponential in $O \big(\sum_{f_i \in \mathcal{F}} |T^{(i)}| \big)$.
\end{theorem}

\begin{proof}
    See Appendix~\ref{app:complexityproof}.
\end{proof}

\subsection{Delay Deficit}

From Theorem~\ref{th:generalcomplexity}, it is clear that searching over all policies in $\Pi_c$ is intractable. To narrow our search, it will prove helpful to identify a tighter bound on packet delays as a function of the scheduling policy. Using a similar argument to that in Lemma~\ref{lemma:taulowerbound}, which gives a lower bound on $\tau_i^*(\pi)$ in terms of minimum inter-scheduling times, we will show an upper bound on $\tau_i^*(\pi)$ subject to conditions on \textit{maximum} inter-scheduling times.

To do so, we introduce a concept called delay deficit. Intuitively, delay deficit provides a quota of worst-case delay which a packet should expect to see at a link under some scheduling assumptions, and it tracks how long each packet has been in the network relative to this quota. Define the age of packet $l$, i.e., the time since it arrived in the network, at time $t$ as $a^l(t)$, and without loss of generality let packet $l$ belong to $f_i$. Denote the delay quota for link $e$ under policy $\pi \in \Pi_c$ as $\sigma_e^{\pi}$.

The delay deficit of a packet $l$ currently enqueued at link $T^{(i)}_j$ at time $t$ is defined as 
\begin{equation}\label{eq:delaydeficit}
    \delta^l(t) \triangleq a^l(t) - \sum_{j'=0}^{j-1} \sigma_{i,j'}^{\pi},
\end{equation}
where we slightly abuse notation to let $\sigma_{i,j}$ denote the delay quota of $T^{(i)}_j$. This quantity evolves in the following way. When packet $l$ arrives at the source, $\delta^l(t) = 0$. Then, it is incremented by one in each subsequent slot until it reaches its destination, and it is decremented by $\sigma_e^{\pi}$ whenever it is served at link $e$. A negative delay deficit signifies that a packet has spent less than its allocated time on its route thus far, so it is ahead of schedule. A delay deficit larger than the delay quota at a packet's current link signifies that the packet is behind schedule.

By induction on the delay deficit at each link, a packet's end-to-end delay can be bounded by the sum of delay quotas along its route, under certain conditions on slice widths. In particular, it leads to the following bound.

\begin{theorem}\label{th:maindelaydefbound}
    Under any policy $\pi \in \Pi_c$, with slice widths $w_{i,e} \geq \lambda_i \overline{k}_e^{\pi}$ for all $f_i \in \mathcal{F}$ and $e \in T^{(i)}$,
    \begin{equation}\label{eq:tauupperbound}
        \tau_i^*(\pi,\boldsymbol{w},\lambda_i) \leq \sum_{e \in T^{(i)}} \overline{k}_e^{\pi}, \ \forall \ f_i \in \mathcal{F},
    \end{equation}
    where $\overline{k}_e^{\pi}$ is the maximum inter-scheduling time of link $e$ under $\pi$.
\end{theorem}

\begin{proof}
    Define the delay quota at each link $e$ as $\sigma_e^{\pi} = \overline{k}_e^{\pi}$. Then the bound in~\eqref{eq:tauupperbound} becomes $\tau_i^*(\pi,\boldsymbol{w},\lambda_i) \leq \sum_{e \in T^{(i)}} \sigma_{i,e}^{\pi}$.

    Under these delay quotas, we claim that if a packet $l$ arrives at link $e$ at time $t$ and $\delta^l(t) \leq 0$, then at whatever time $t'$ it arrives at link $e+1$, $\delta^l(t') \leq 0$. To see this, assume each queue has a counter which tracks the number of packets with strictly positive delay deficit. Our assumption is that packets arrive at link $e$ with a negative delay deficit, so they are not added to the counter on arrival. From~\eqref{eq:delaydeficit}, a strictly positive delay deficit implies that a packet's age $a^l(t) > \sum_{j=0}^{e-1} \sigma_{i,j}^{\pi}$ at time $t$, and because $\lambda_i$ packets arrive at the source link in each slot, at most $\lambda_i$ packets can reach this age and be added to the counter in each slot.

    Now we claim that when any link $e$ is scheduled, it serves all packets with positive delay deficits, and will show this by induction. Assume that link $e$ is scheduled at time $t$ and it serves all packets with $\delta(t) > 0$. Then the next time it is scheduled at time $t'$, it again must serve all packets with $\delta(t') > 0$. To see this, recall that link $e$ is scheduled at least every $\overline{k}_e^{\pi}$ slots by definition, so $t' \leq t+\overline{k}_e^{\pi}$. Then at most $\lambda_i \overline{k}_e^{\pi}$ packets can have a positive delay deficit at $t'$, and because $w_{i,e} = \lambda_i \overline{k}_e^{\pi}$, all these packets are served. The first time that link $e$ is scheduled, there can be no more than $\lambda_i \overline{k}_e^{\pi}$ packets in the network, so it necessarily serves all packets with positive delay deficits, and this completes the induction step.

    Therefore, every packet $l$ is served the first time $t'-1$ that link $e$ is scheduled after $\delta^l$ becomes positive, so $\delta^l(t'-1) < \overline{k}_e^{\pi}$ when it is served. This is exactly the delay quota of link $e$, so when packet $l$ arrives at the next link at time $t'$, $\delta^l(t') \leq 0$. By definition, packets arrive at the source link with $\delta^l = 0$, which completes the induction step. Now consider a ficticious exit link at the destination of each flow. Each packet must arrive at this link with a delay deficit at most $0$, which means its age is at most $\sum_{e \in T^{(i)}} \sigma_{i,e}^{\pi} = \sum_{e \in T^{(i)}} \overline{k}_e^{\pi}$. This completes the proof.
\end{proof}

We note several things about this result. First, it subsumes the worst-case delay bound under primary interference in Theorem~\ref{th:deadlinebound}. In the proof of that theorem, the worst-case scenario is described to have a maximum inter-scheduling time $\overline{k}_e^{\pi} = K^{\pi}(1-\bar{\mu}_e^{\pi})$ for all links $e$, which makes the bounds identical up to a difference of $1$ slot. Second, when links are scheduled more regularly and $\overline{k}_e$ is independent of the schedule length, this bound can be significantly tighter than that in Theorem~\ref{th:deadlinebound}. In fact, it is a constant factor of $(\sum_{e \in T^{(i)}} \overline{k}_e^{\pi}) / |T^{(i)}|$ from the deadline-optimal lower bound on $\tau_i^*$ for a solitary flow in Theorem~\ref{th:orrbound}.

A final thing to note is that Theorem~\ref{th:maindelaydefbound} requires slice widths that are not resource-minimizing in general. Define the additional capacity required beyond the resource-minimizing value as
\begin{equation}\label{eq:deltaw}
    \Delta w_{i,e}(\pi,\lambda_i) \triangleq \lambda_i \big( \overline{k}_e^{\pi} - \frac{1}{\bar{\mu}_e^{\pi}} \big) \geq 0, \ \forall \ f_i \in \mathcal{F}, \ e \in T^{(i)}.
\end{equation}
Because the average inter-scheduling time is $1/\bar{\mu}_e^{\pi}$, the quantity $\Delta w_{i,e}(\pi,\lambda_i)$ is small when inter-scheduling times are somewhat regular and $\overline{k}_e^{\pi}$ is not much larger than this average.

It is important to clarify that, with these conditions on slice widths, a link is not guaranteed to empty its queue each time it is scheduled. Any given packet can spend more than $\overline{k}_e^{\pi}$ slots at link $e$, but Theorem~\ref{th:maindelaydefbound} guarantees it will make up for this by spending fewer than $\overline{k}_{e'}^{\pi}$ slots at some other link $e'$ on its route. The strength of Theorem~\ref{th:maindelaydefbound} lies in this fact, which enables us to keep $\Delta w_{i,e}(\pi,\lambda_i)$ small and to allocate slices efficiently.

From~\eqref{eq:tauupperbound} and~\eqref{eq:deltaw}, we see that minimizing maximum inter-scheduling times leads to both tighter delay bounds and more efficient slices. If inter-scheduling times are exactly equal, and $\underline{k}_e^{\pi} = \overline{k}_e^{\pi} = \frac{1}{\bar{\mu}_e^{\pi}}$ for all $e \in E$, then $\overline{k}_e^{\pi}$ is by definition minimized for a fixed $\bar{\mu}_e^{\pi}$. If this holds, we say that $\pi$ is a \textit{regular schedule}.

\begin{corollary}\label{cor:regularschedules}
    A policy $\pi \in \Pi_c$ supports a set of flows $\mathcal{F}$ with slice widths $w_{i,e} = \lambda_i \overline{k}_e^{\pi}$ for all $f_i \in \mathcal{F}$ and $e \in T^{(i)}$, if 
    \begin{align}\label{eq:rateconditions}
    \begin{aligned}
        &\sum_{e \in T^{(i)}} \overline{k}_e^{\pi} \leq \tau_i, \ \forall \ f_i \in \mathcal{F}, \\
        &\sum_{f_i : e \in T^{(i)}} \lambda_i \overline{k}_e^{\pi} \leq c_e, \ \forall \ e \in E. \\
    \end{aligned}
    \end{align}
    Furthermore, if $\pi$ is a regular schedule, then slices are resource-minimizing.
\end{corollary}

\begin{proof}
    The sufficient conditions come from the deadline bound in Theorem~\ref{th:maindelaydefbound} and link capacity constraints. If both of these are satisfied, then $\pi$ is guaranteed to support $\mathcal{F}$ by Theorem~\ref{th:maindelaydefbound}. To show the second part, note that by definition,~\eqref{eq:tauupperbound} holds with $\Delta w_{i,e}(\pi,\lambda_i) = 0$ if and only if $\overline{k}_e^{\pi} = \frac{1}{\bar{\mu}_e^{\pi}}$ for all $e \in E$, which is true only when $\pi$ is regular.
\end{proof}

While the bound in Theorem~\ref{th:maindelaydefbound} is not tight in general, and smaller deadline guarantees may be possible, it highlights the importance of keeping inter-scheduling times small. This intuitively makes sense given our knowledge of the deadline-optimal ORR policy. The number of interfering links and the multi-directional flows in this general setting make it impossible to schedule links in order along a flow's route, as in the ORR policy. Rather, by bounding the inter-scheduling times, we can minimize how far the schedule deviates from this order, and do this simultaneously for all flows. This keeps end-to-end delays from growing too large, and allows us to tune the inter-scheduling times on some flows' routes to become ``closer'' to the ORR policy when subject to tighter deadlines. In fact, setting $\overline{k}_e = 2$ for all links on a flow's route recovers the ORR policy exactly under primary interference.

Motivated by this, we develop efficient algorithms in the next section to construct schedules with bounded inter-scheduling times, focusing particularly on regular schedules. We will see that by using the sufficient conditions in Corollary~\ref{cor:regularschedules}, we are able to develop policies which support $\mathcal{F}$ in polynomial time. Because the conditions in Corollary~\ref{cor:regularschedules} are not necessary, this does not violate Theorem~\ref{th:generalcomplexity}, but rather accomplishes our goal of narrowing the policy search to a smaller subclass of $\Pi_c$, trading some performance for a large reduction in complexity.

\section{Algorithm Development}\label{sec:linkscheduling}

In this section, we continue to focus on primary interference, for which activation sets are matchings on the graph $G$. We have seen that regular schedules allow for resource-minimizing slices, which maximize network capacity, so we seek to find regular schedules whenever possible. Unfortunately such schedules often do not exist for a given set of activation rates (we will show this later), so we allow for \textit{almost-regular} schedules, where inter-scheduling times can differ by at most one slot. For ease of notation, when speaking of regular or almost-regular schedules, we denote the maximum inter-scheduling time $\overline{k}_e$ as simply $k_e$. The remainder of the section details the steps our algorithm takes to construct these schedules.

\subsection{Initialization}

Our algorithm takes as input a set of flows $\mathcal{F}$, and throughput and deadline constraints $\lambda_i$ and $\tau_i$ respectively for each $f_i \in \mathcal{F}$. Assume we are able to design an almost-regular schedule, where link $e$ has $\alpha$ inter-scheduling times of $k_e$ and $\beta$ inter-scheduling times of $k_e - 1$. Then for any schedule length $K$, $\alpha k_e + \beta (k_e - 1) = K$. Recall also that $\bar{\mu}_e = \eta_e/K = (\alpha + \beta)/K$, so by rearranging terms and substituting,
\begin{equation}\label{eq:kceiling}
    k_e = \frac{K}{\alpha+\beta} + \frac{\beta}{\alpha+\beta} = \frac{1}{\bar{\mu}_e} + \frac{\beta}{\alpha+\beta} = \Big\lceil \frac{1}{\bar{\mu}_e} \Big\rceil < \frac{1}{\bar{\mu}_e}+1,
\end{equation}
for any $\alpha$ and $\beta$, because $k_e$ is integer-valued. Note that the last equality holds because $\beta/(\alpha+\beta) < 1$ and $1/\bar{\mu}_e$ is integer if and only if the schedule is regular and $\beta = 0$.

From Corollary~\ref{cor:regularschedules} and the bound in~\eqref{eq:kceiling}, any set of activation rates which solve the following program meet service guarantees, provided the algorithm can construct an almost-regular schedule satisfying these rates.
\begin{align}\label{eq:initiallinkactivations}
\begin{aligned}
    \min \ &\sum_{e \in E} \bar{\mu}_e \\
    \text{s.t.} \ &\sum_{e \in T^{(i)}} \Big(\frac{1}{\bar{\mu}_e} + 1 \Big) \leq \tau_i, \ \forall \ f_i \in \mathcal{F}, \\
    &\sum_{f_i : e \in T^{(i)}} \lambda_i \Big(\frac{1}{\bar{\mu}_e} + 1 \Big) \leq c_e, \ \forall \ e \in E.
\end{aligned}
\end{align}
This problem is convex, and we denote the solution as $\bar{\mu}^0$, also referred to as the initial link rates. The algorithm must now find an almost-regular schedule where $\bar{\mu}_e \geq \bar{\mu}_e^0$ for all links $e$.

\subsection{Unique-Edge Matchings}

One hurdle in finding schedules with regularity constraints under primary interference is that, in general, work-conserving policies assign links to more than one matching~\cite{hajek_link_1988}. As a result, constructing a regular schedule of matchings does not guarantee the schedule is regular for each link, which complicates the problem.

Finding an almost-regular schedule that satisfies the constraints in~\eqref{eq:initiallinkactivations}, with fixed values of $k_e$ and a fixed schedule length, can be formulated as a Periodic Event Scheduling Problem (PESP), which is well-studied and known to be NP-complete~\cite{serafini1989mathematical,dauscha_cyclic_1985}
~\footnote{In the conference version of this paper, we incorrectly claim that our problem is NP-complete by equivalence to PESP. We clarify here that solving PESP is the best known method of solving our problem when these parameters are fixed, but they are not equivalent problems. In particular, not all versions of PESP can be formulated as an almost-regular scheduling problem.}. More generally, and without fixing the schedule length, our problem can be solved by finding a cycle on a graph $\tilde{G} = (\tilde{V}, \tilde{E})$. Each $v \in \tilde{V}$ represents a state $(\hat{t}_1,\hat{t}_2,\dots,\hat{t}_{|E|})$, such that $\hat{t}_e$ is the time since link $e$ was most recently scheduled, and edges in $\tilde{E}$ represent valid matchings, which also satisfy almost-regularity constraints. In other words, a matching which includes link $e$ is only valid in a state $v$ where $\hat{t}_e \in \{ k_e-1, k_e \}$ for fixed values of $k_e$. Any cycle in this graph represents a valid almost-regular schedule.



Because $\bar{\mu}_e \geq \bar{\mu}_e^0$, we have from~\eqref{eq:kceiling} that $k_e - 1 \leq \frac{1}{\bar{\mu}_e^0}$ for all $e$, so the number of combinations of $k_e$ that satisfy~\eqref{eq:initiallinkactivations} is bounded and exponential in $|E|$. In the worst case, we must consider this many different graphs to find a cycle. Likewise, because $\hat{t}_e \leq k_e$ for any valid state, the number of states $|\tilde{V}|$ in any $\tilde{G}$ is bounded and exponential in $|E|$. From~\cite[Sec.~22]{cormen2022introduction}, the most efficient known algorithm to detect a cycle in a graph is DFS, which has complexity $O(|\tilde{V}|+|\tilde{E}|)$. There is at least one valid action in each state, so the overall complexity is at least $O(\beta^{2|E|})$ for some constant $\beta$.

It may seem that we are back to square one, given that solving this problem requires exponential time, just like the problem we started with in Theorem~\ref{th:generalcomplexity}. By further constraining links to belong to a single matching, however, we can make the problem tractable. Under this additional constraint, if matchings are scheduled regularly, then each link is also scheduled regularly. We refer to such matchings as \textit{unique-edge matchings}, and will see that we can find almost-regular schedules of such matchings in polynomial time.


Clearly each link in a unique-edge matching is activated at the same rate as its matching. In order to ensure the constraints in~\eqref{eq:initiallinkactivations} are met, our algorithm must guarantee the initial matching rate $\bar{\mu}_m^0 \geq \bar{\mu}_e^0$ for all links $e$ belonging to matching $m$, and for all matchings $m$. Then the unique-edge matching problem can be written as
\begin{align}\label{eq:edgematchings}
\begin{aligned}
    \min_{M \subseteq M^1} &\sum_{m \in M} \bar{\mu}_m^0 \\
    \text{s.t.} \ &\bar{\mu}_m^0 \geq \bar{\mu}_e^0, \ \forall \ e \in m, m \in M\\
        &e \in M, \ \forall \ e \in E, \\
        &m \cap m' = \varnothing, \ \forall \ m, m' \in M,
\end{aligned}
\end{align}
where $M^1$ is the set of all feasible matchings and $\varnothing$ is the empty set. Because only one matching can be activated at a time, $M$ is feasible if and only if $\sum_{m \in M} \bar{\mu}_m^0 \leq 1$, which motivates the objective function.

Despite being of practical interest, and related to the minimum edge sum coloring problem~\cite{bar-noy_chromatic_1998} and the weighted sum coloring problem~\cite{bar2000sum, epstein_weighted_2009}, this exact formulation has received no attention in the literature. We define a \textit{Greedy Matching} (GM) algorithm, which has complexity $O(|E|^2)$ and provides an approximate solution to this problem, in Algorithm~\ref{alg:greedymatching}. 
\begin{algorithm}
    \DontPrintSemicolon
    %
    \SetKwInput{Input}{Input}\SetKwInOut{Output}{Output}
    \Input{Initial link activation rates $\bar{\mu}_e^0, \ \forall \ e \in E$}
    \Output{Set of unique edge matchings $M^{GM}$}
    Sort links in $E$ by $\bar{\mu}_e^0$, largest to smallest \\
    Define matching set $M^{GM}$ \\
    \While{$E$ is not empty}{
        Add new matching $m$ to $M^{GM}$ \\
        Add first link $e' \in E$ to $m$ and remove $e'$ from $E$ \\
        Set $\bar{\mu}_m^0 = \bar{\mu}_{e'}^0$ \\
        \For{$e \in E$}{
            \If{$m \cup e$ is a valid matching}{
                Add $e$ to $m$ and remove $e$ from $E$ \\
            }
        }
    }

    Return $M^{GM}, \ \bar{\mu}_m^0, \ \forall \ m \in M^{GM}$ \\   
    
    \caption{Greedy Matching (GM)}
    \label{alg:greedymatching}
\end{algorithm}

The algorithm sorts links in decreasing order of initial rates, and greedily adds each link to a matching if it does not interfere with any links previously added. Once all links have been considered, the matching is full and is assigned an initial rate equal to that of its first link, which by construction is the largest. This greedy heuristic ensures the constraints in~\eqref{eq:edgematchings} are met and the matchings it returns are unique-edge matchings, while keeping the objective within a small factor of optimal, as shown next.
\begin{theorem}\label{th:gdmapproximation}
    Let $M^*$ be a solution (i.e., an optimal set of matchings) to~\eqref{eq:edgematchings}. The GM algorithm produces a set of matchings $M^{GM}$ such that 
    \begin{equation}
        \sum_{m \in M^{GM}} \bar{\mu}_m^0 \leq \Big( \frac{|M_{GM}|}{\Delta(G)} \Big) \sum_{m \in M^*} \bar{\mu}_m^0,
    \end{equation}
    where $\Delta(G)$ is the maximum degree of $G$. This is a factor of $2$ from optimal in the worst case, provided $\Delta(G) = O(\log |V|)$.
\end{theorem}

\begin{proof}
    See Appendix~\ref{app:gmapproxproof}.
\end{proof}

\subsection{Feasibility Conditions}

Next we analyze the feasibility conditions for both regular and almost-regular schedules of unique-edge matchings. Define a \textit{step-down} vector as a sorted vector of values, where each element is an integer multiple of the next. When applied to activation rates, this property turns out to be quite useful in constructing regular schedules.
\begin{lemma}\label{lemma:regconditions}
    A regular schedule of unique-edge matchings exists if
    \begin{enumerate}
        \item $\sum_{m \in M^{GM}} \bar{\mu}_m = 1$,
        \item the vector of rates $\bar{\mu}$ is step-down,
        \item and $k_m = \lceil \frac{1}{\bar{\mu}_m} \rceil = \frac{1}{\bar{\mu}_m}$ for all $m$.
    \end{enumerate}    
    Furthermore, an almost-regular schedule exists if only conditions $(1)$ and $(2)$ hold.
\end{lemma}

\begin{proof}
    See Appendix~\ref{app:regconditionsproof}.
\end{proof}

Initial rates can always be increased without violating the conditions in~\eqref{eq:rateconditions}, so if a set of rates is feasible, then the rates can be normalized to sum to $1$ in order to meet condition $(1)$. The other conditions are more difficult to meet, particularly at the same time, and a regular schedule is not guaranteed to exist for a feasible set of rates.

In~\cite{li_scheduling_2021}, Li et al. construct a polynomial-time algorithm which takes a set of rates and outputs a set of possibly larger \textit{augmented rates} $\hat{\mu}$ that are step-down. We refer to this as the Augmented Rates algorithm. When applied to the initial activation rates of unique-edge matchings, the output satisfies condition $(2)$ from Lemma~\ref{lemma:regconditions}. We do not replicate the algorithm here due to space constraints, but encourage readers to reference Algorithm 2 and the accompanying results in~\cite{li_scheduling_2021}.

Despite satisfying condition $(2)$, these augmented rates do not always satisfy condition $(3)$, and there is no guarantee that we can form a regular schedule. However, when $\sum_m \hat{\mu}_m \leq 1$, the rates can be normalized to meet condition $(1)$ as described above, without affecting the step-down property of the vector. This ensures that conditions $(1)$ and $(2)$ are met and we can form an almost-regular schedule. Extrapolating a feasibility result from~\cite{li_scheduling_2021} to our setting yields the following sufficient condition for this to hold.
\begin{lemma}\label{lemma:almostregcond}
    An almost-regular schedule of unique-edge matchings exists if the sum of initial rates $\sum_m \bar{\mu}_m^0 \leq \ln 2 \approx 0.69$.
\end{lemma}

\begin{proof}
    The proof follows from~\cite{li_scheduling_2021}.
\end{proof}

\subsection{Schedule Construction}

Lemma~\ref{lemma:almostregcond} shows that when initial rates are small enough, the Augmented Rates algorithm produces a vector of rates that are step-down and can be normalized to sum to $1$. We can construct an almost-regular schedule from these rates by first finding a longer, regular schedule where not all slots are filled, and then removing the unused slots. The algorithm for constructing schedules in this manner is likewise first described by Li et al. in~\cite{li_scheduling_2021}. We replicate and briefly describe the algorithm here. For a full proof of correctness, which was omitted in the original work, see the proof by construction of Lemma~\ref{lemma:regconditions} in Appendix~\ref{app:regconditionsproof}.
\begin{algorithm}
    \DontPrintSemicolon
    %
    \SetKwInput{Input}{Input}\SetKwInOut{Output}{Output}
    \Input{Flows $\mathcal{F}$, throughput/deadline constraints $(\lambda_i, \tau_i)$}
    \Output{Almost regular schedule $\pi \in \Pi_c$}
    Find initial link activations $\bar{\mu}^0_e, \ \forall \ e$, by solving~\eqref{eq:initiallinkactivations} \\
    Find unique-edge matchings $M^{GM}$ and initial rates $\bar{\mu}^0_m, \ \forall \ m$, using Greedy Matching \\
    Find $\hat{\mu}$ using the Augmented Rates algorithm from~\cite{li_scheduling_2021} \\
    \If{$\sum_m \hat{\mu}_m > 1$}{
        Return None \\
    }
    Normalize rates $\bar{\mu}_m^{\pi} = \hat{\mu}_m / \sum_m \hat{\mu}_m$ for all $m$ \\
    Set $M = |M^{GM}|$, $K^{\pi} = 1 / \bar{\mu}_M^{\pi}$ and $\eta_m = \bar{\mu}_m^{\pi} K^{\pi}$ for all $m$ \\
    Form an empty schedule with length $K' = \lceil \frac{1}{\bar{\mu}_1^{\pi}} \rceil \eta_1$. \\
    Assign slot $1$ to $m_1$, along with every $K'/\eta_1$ subsequent slots \\
    \For{$i=2,3,\dots,|M|$}{
        Set $S_0$ as the set of empty slots \\
        \For{$j=1,2,\dots,i-1$}{
            Set $S_j \subseteq S_{j-1}$ as the set of slots most closely following a slot assigned to $m_j$ \\ 
        }
        Assign the first slot in $S_{i-1}$ to $m_i$, along with every $K'/\eta_i$ subsequent slots \\
    }
    Remove the $K'-K^{\pi}$ unassigned slots from the schedule \\
    Set $k_e^{\pi}$ as the max inter-scheduling time of link $e$ in $\pi$ \\
    Set $w_{i,e} = \lambda_i k_e^{\pi}$ for all $f_i \in \mathcal{F}$ and $e \in T^{(i)}$ \\
    
    Return $\pi, \ w$
    \caption{Almost-Regular Schedule Construction}
    \label{alg:schedconstr}
\end{algorithm}

Let $\bar{\mu}$ be the augmented vector of rates after normalizing, sorted from largest to smallest, $M$ be the number of matchings, and $K = k_M = 1/\bar{\mu}_M$ be the schedule length. From the step-down property, this is guaranteed to be an integer. To see this, let $z_{ij} = \frac{\bar{\mu}_i}{\bar{\mu}_j}$ for all pairs of matchings $i$ and $j$, and note that $z_{ij}$ is integer-valued by the step-down property when $i \leq j$. Then $\bar{\mu}_m = z_{m,M} \bar{\mu}_M$ for all $m$. Substituting into condition $(1)$ of Lemma~\ref{lemma:regconditions} and rearranging yields
\begin{equation}
    K = \frac{1}{\bar{\mu}_M} = \sum_m z_{mM},
\end{equation}
which by definition is an integer.

Define the number of occurrences of matching $m$ in the schedule as $\eta_m = \bar{\mu}_m K$, and define an auxiliary schedule of length $K' = k_1 \eta_1 = \lceil \frac{1}{\bar{\mu}_1} \rceil \eta_1$. Note that the schedule length is an integer multiple of $k_1$, which allows matching $1$ to be scheduled regularly every $k_1$ slots when it does not overlap or ``collide'' with any other matching in the schedule. Now consider an arbitrary matching $m$. It too can be scheduled regularly every $k_m' = \frac{K'}{\eta_m}$ slots (assuming no collisions) if $k_m'$ is an integer. By definition,
\begin{equation}
    k_m' = \frac{K'}{\bar{\mu}_m K} = \frac{z_{1m} K'}{\bar{\mu}_1 K} = \frac{z_{1m} K'}{\eta_1} = z_{1m} k_1,
\end{equation}
which is integer-valued. Therefore, every matching can be scheduled regularly in a schedule of length $K'$, provided there are no collisions in the schedule. Collisions can be avoided by scheduling matchings in a greedy fashion, beginning with matching $1$ and proceeding in increasing order of $k_m'$. Because the vector is already sorted, this is simply the ordering of the vector. For each matching $m$, pick an empty slot in the schedule and assign it to $m$, along with every $k_m'$ subsequent slots. Due to the step-down property of the vector and the greedy ordering, it is easy to show that no collisons can occur following this method.

The resulting schedule is regular for all matchings, and has length $K'$. If $\frac{1}{\bar{\mu}_1}$ is integer-valued, then $K' = K$, and the method described above generates a regular schedule that supports $\mathcal{F}$. If not, there will be $K'-K$ empty slots in the schedule after all matchings have been added. These slots are then removed, yielding a schedule of length $K$ with activation rates $\bar{\mu}_m$ for each matching $m$. To ensure this schedule is almost-regular, the first slot assigned to matching $m$ should be picked to spread out the remaining empty slots.

The full details of the schedule construction are described in Algorithm~\ref{alg:schedconstr}, which we refer to as the Almost-Regular Schedule Construction (ARSC) algorithm. This algorithm has complexity $O(|E|^2 + M^2 K^2 + X)$, where $X$ is the complexity of the algorithm used to solve the convex program~\eqref{eq:initiallinkactivations}.
\begin{theorem}
    If the augmented rates found in step $3$ of the ARSC algorithm satisfy $\sum_m \hat{\mu}_m \leq 1$, then the algorithm returns an almost-regular schedule $\pi$ which supports $\mathcal{F}$, and a corresponding set of slices, such that each slice $w_{i,e}$ is within a factor of $\frac{1}{k_e^{\pi}} \leq \bar{\mu}_e^{\pi}$ of being resource-minimizing. 
\end{theorem}

\begin{proof}
    To show that $\pi$ supports $\mathcal{F}$, one needs only to show that it satisfies the sufficient conditions in Corollary~\ref{cor:regularschedules}. From the constraints in the Greedy Matching and Augmented Rates algorithms, $\hat{\mu}_m \geq \bar{\mu}_m^0$ for all $m$, so this is guaranteed to hold.

    To show the second part, note that the algorithm sets slice widths equal to $w_{i,e} = \lambda_i k_e^{\pi}$, and because the schedule is almost-regular and links are scheduled every $k_e^{\pi}$ or $k_e^{\pi}-1$ slots, $\frac{1}{\bar{\mu}_e^{\pi}} \geq k_e^{\pi}-1$ Then, from~\eqref{eq:deltaw}, the excess slice width is 
    \begin{equation}
        \Delta w_{i,e}(\pi,\lambda_i) \leq \lambda_i (k_e^{\pi} - (k_e^{\pi}-1)) = \lambda_i,
    \end{equation}
    and the fraction of excess slice width is 
    \begin{equation}
        \frac{\Delta w_{i,e}(\pi,\lambda_i)}{w_{i,e}} \leq \frac{\lambda_i}{\lambda_i k_e^{\pi}} = \frac{1}{k_e^{\pi}} \leq \bar{\mu}_e^{\pi}, \ \forall \ e \in E.
    \end{equation}
\end{proof}

The final part of this theorem is critical, because it shows that our algorithm does not waste network resources to meet service guarantees. A naive method of meeting guarantees is to overprovison slices, but this result shows that our algorithm does not rely on this, and that a network can operate at near capacity while implementing ARSC.

We further illustrate the schedule construction and details of the ARSC algorithm through an example. Let the normalized augmented rates be the step-down vector $\bar{\mu} = \big(\frac{2}{5}, \frac{1}{5}, \frac{1}{5}, \frac{1}{10}, \frac{1}{10})$, the schedule length $K = 10$, and $\eta_1 = 4$. We begin by constructing a longer schedule of length $K' = \eta_1 k_1 = 12$ slots, and assign matching $1$ to the first slot followed by every $k_1 = 3$ subsequent slots. This gives an initial schedule
\begin{equation*}
    \underline{1} \ \underline{\phantom{x}} \ \underline{\phantom{x}} \ \underline{1} \ \underline{\phantom{x}} \ \underline{\phantom{x}} \ \underline{1} \ \underline{\phantom{x}} \ \underline{\phantom{x}} \ \underline{1} \ \underline{\phantom{x}} \ \underline{\phantom{x}}.
\end{equation*}

For each subsequent matching $i$, the algorithm takes the set of empty slots $\chi_1$ that most closely follow matching $1$ in the schedule (here $\chi_1 = \{2,5,8,11\}$). Note that because the schedule is cyclic, we assume it wraps around when computing the closest following slots. If $i > 2$, it then forms a set $\chi_2 \subseteq \chi_1$ with the set of empty slots that most closely follow matching $2$. It proceeds through each matching which has already been scheduled in the same fashion, and finally assigns the first slot in $\chi_{i-1}$ to matching $i$, followed by slots at subsequent spacings of $K/\eta_i$. Following this method, the first slot in $\chi_1$ is slot $2$, so matching $2$ is assigned to this slot, followed by slots at equal spacings of $K/\eta_2 = 6$. This yields 
\begin{equation*}
    \underline{1} \ \underline{2} \ \underline{\phantom{x}} \ \underline{1} \ \underline{\phantom{x}} \ \underline{\phantom{x}} \ \underline{1} \ \underline{2} \ \underline{\phantom{x}} \ \underline{1} \ \underline{\phantom{x}} \ \underline{\phantom{x}}.
\end{equation*}

Moving on to matching $3$, we have $\chi_1 = \{5,11\}$. Both of these slots lag matching $2$ by $3$ slots in the schedule, so $\chi_2 = \chi_1$, and we assign matching $3$ to the first slot $5$, followed by slots at increments of $K/\eta_3 = 6$. This gives us
\begin{equation*}
    \underline{1} \ \underline{2} \ \underline{\phantom{x}} \ \underline{1} \ \underline{3} \ \underline{\phantom{x}} \ \underline{1} \ \underline{2} \ \underline{\phantom{x}} \ \underline{1} \ \underline{3} \ \underline{\phantom{x}}.
\end{equation*}

For matching $4$, $\chi_1 = \{3,6,9,12\}$, and $\chi_2 = \chi_3 = \{3,9\}$. Therefore we assign matching $4$ to slot $3$, and because $\eta_4 = 12$, this is the only slot where it is assigned. Finally, for matching $5$, $\chi_1 = \{6,9,12\}$, and $\chi_2 = \chi_3 = \chi_4 = 9$. Adding matching $5$ and removing the remaining empty slots ultimately yields the schedule
\begin{equation*}
    \underline{1} \ \underline{2} \ \underline{4} \ \underline{1} \ \underline{3} \ \underline{1} \ \underline{2} \ \underline{5} \ \underline{1} \ \underline{3}.
\end{equation*}
One can verify that this final schedule has length $K=10$, the activation rates are equal to $\bar{\mu}$, and the schedule is almost-regular.

\subsection{Discussion}

The ARSC algorithm provides an efficient, polynomial-time method for constructing schedules to meet hard throughput and deadline constraints, under primary interference and general network topologies. It can meet tight deadline guarantees, and it allocates slices within a small factor of being resource-minimizing. Of course, Theorem~\ref{th:generalcomplexity} shows that finding a feasible policy in general suffers from exponential complexity, so ARSC cannot find a schedule for every set of flows where one exists. The algorithm runs in polynomial time by constraining the problem and taking suboptimal steps. 

Forcing the class of policies to be regular or almost-regular is one such restriction, as are the unique-edge matching condition and the requirement that augmented rate vectors be step-down. These conditions are necessary for the ARSC algorithm to produce a feasible schedule, but are not necessary for a feasible schedule to exist. The gap between the set of all feasible schedules and the set returned by ARSC is the price we pay for the reduction in complexity. Lemma~\ref{lemma:almostregcond} shows that when the sum of initial matching rates is less than $\ln 2 \approx 0.69$, ARSC will always succeed, but there is no necessary condition on the algorithm's success that we have found. We explore the algorithm's feasible solution rate through simulations in the next section.

\section{Numerical Results}\label{sec:simulations}

We support our results in this section by simulating the ARSC algorithm and observing its performance numerically. In all our simulations we use the example network in Figure~\ref{fig:network-diagram}, and note that while the diagram shows bidirectional links, we treat each of these as two directional links which interfere as in the analysis. We further assume the network topology and link capacities are fixed. 

In each experiment, we generate $32$ flows with randomly chosen source/destination pairs, and fix their routes using shortest path routing. We assume that each flow has an identical throughput $\lambda$ and deadline $\tau$. For perspective on our algorithm's performance, we compute the largest value $\lambda^*$ that $\lambda$ can take independent of any deadline constraints, by finding a throughput-optimal set of matchings. We then scale the throughput of each flow using this value.
\begin{figure}
    \centering
    \includegraphics[width=0.3\textwidth]{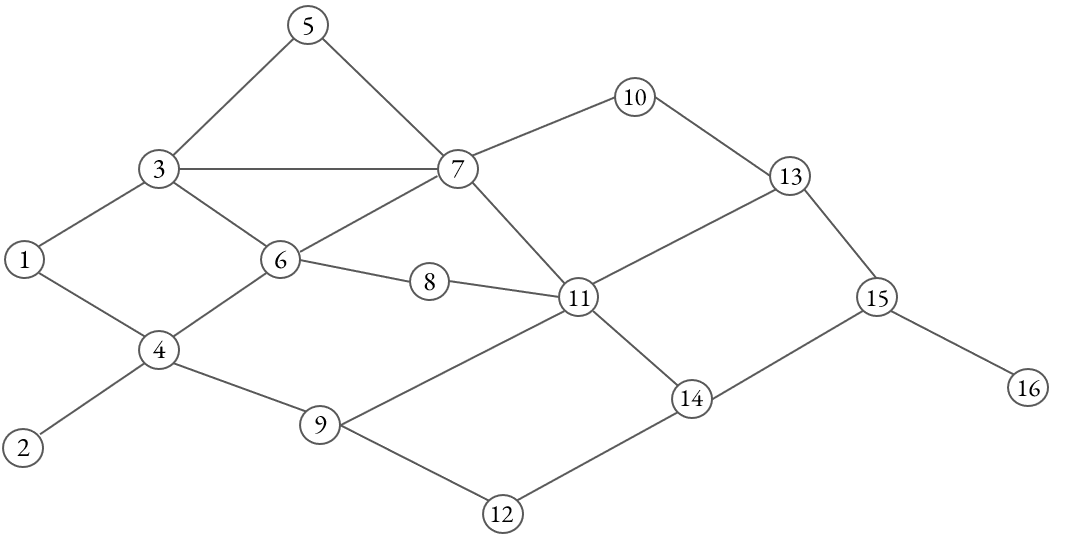}
    \caption{Example Network Diagram}
    \label{fig:network-diagram}
\end{figure}
\begin{figure}
    \begin{subfigure}{.23\textwidth}
        \centering
        \includegraphics[width=.95\linewidth]{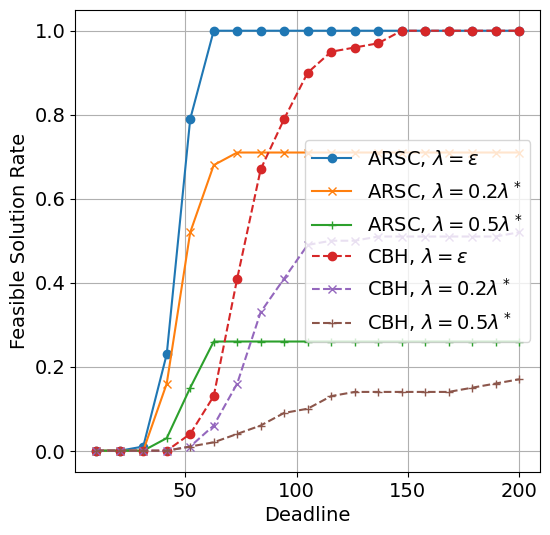}
    \end{subfigure}%
    \begin{subfigure}{.23\textwidth}
        \centering
        \includegraphics[width=.95\linewidth]{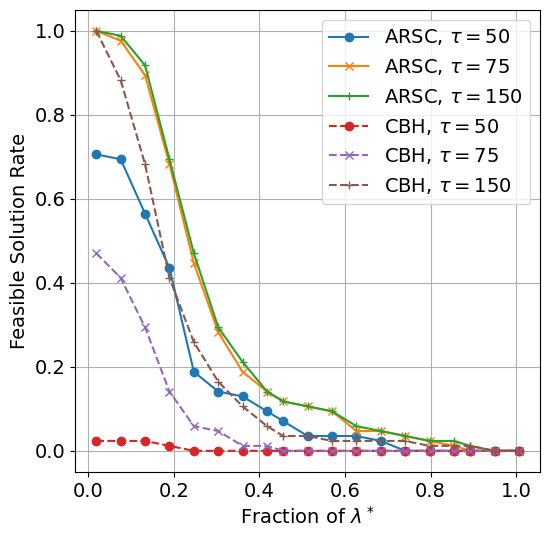}
    \end{subfigure}
    \caption{Feasibility Rate}
    \label{fig:feasibililty-results}
\end{figure}

In the previous section, the ARSC algorithm was shown to meet service guarantees when it returns a feasible almost-regular schedule, so we first examine how often this occurs. In Figure~\ref{fig:feasibililty-results}, we show the feasible solution rate of the ARSC algorithm under varying throughputs $\lambda$ and deadlines $\tau$ by generating $100$ random sets of $32$ flows and averaging the results. We directly compare these results to the Credit-Based Heuristic (CBH) algorithm in~\cite{chilukuri_delay-aware_2015}, the closest to our work and the most recent result in the line of work starting with~\cite{djukic2007quality}, which makes use of contiguous block scheduling and which we have referenced on several occasions.

On the left, we sweep deadlines and run both algorithms under three values of $\lambda$, where $\epsilon$ is the smallest throughput allowed by the implementation, and the other values are scaled by $\lambda^*$. For an arbitrarily small $\lambda = \epsilon$, we expect ARSC to always return a policy for a deadline of $60$ and above, because the longest path any flow can take is $6$ hops, and the network can be colored with $10$ colors. Therefore, a simple round-robin policy will ensure the deadline is met, and because a round-robin schedule is regular by definition, ARSC will find it. The plot confirms that this is indeed the case.

In any direct comparison between ARSC and CBH, ARSC shows significant improvement. This is most notable in the deadline range between $40$ and $70$, during which CBH still has a small solution rate but where ARSC reaches $100\%$ for a throughput of $\epsilon$ and $70\%$ for a throughput of $0.2 \lambda^*$. Below this range, using the same arguments of path length and graph coloring referenced above, the probability that any policy can find a solution for all $32$ flows becomes increasingly low. By the opposite argument, it becomes easier to do so as deadlines increase beyond this range. This region where meeting deadlines is just barely feasible is where our algorithm shines, proving that ARSC can meet tight deadline guarantees. On the right side of Figure~\ref{fig:deadline-results}, we sweep throughput and perform the same comparison for three fixed deadlines, showing more clearly the tradeoff between feasibililty and throughput.
\begin{figure}
    \begin{subfigure}{.23\textwidth}
        \centering
        \includegraphics[width=.95\linewidth]{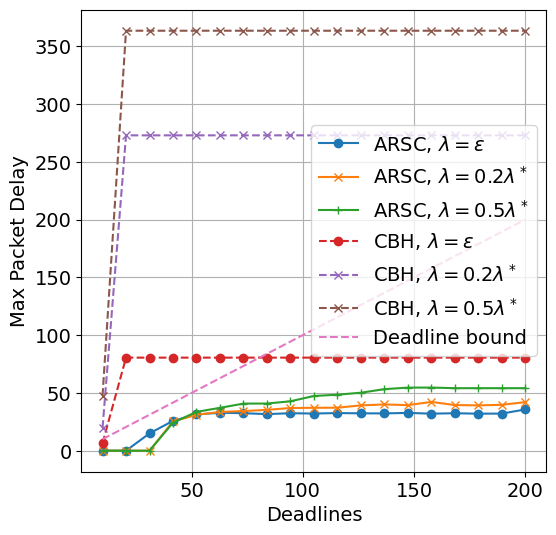}
    \end{subfigure}%
    \begin{subfigure}{.23\textwidth}
        \centering
        \includegraphics[width=.95\linewidth]{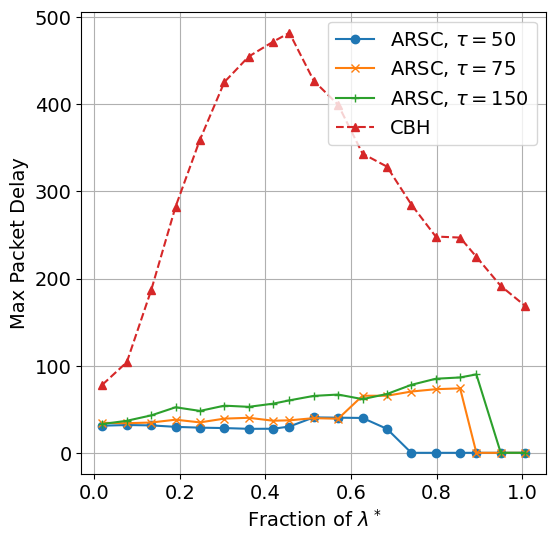}
    \end{subfigure}
    \caption{Max Packet Delay}
    \label{fig:deadline-results}
\end{figure}

In Figure~\ref{fig:deadline-results}, we plot the maximum delay seen by any packet under each simulated policy, averaged over the same $100$ sets of randomly generated flows. On the left, we again sweep deadlines and plot results for both ARSC and CBH under varying throughputs. We also plot the deadline bound for reference. When there is no feasible policy under ARSC, the max delay is plotted as $0$. As expected, no packet sees a delay larger than the deadline bound under ARSC. The difference in delay between throughput levels is relatively small, which makes sense if the limiting factor is interference constraints and schedule order, rather than link capacity.

The CBH algorithm sees much larger packet delays on average, which we expect given that link inter-scheduling times are close to the schedule length. This demonstrates our result that regular schedules lead to much smaller packet delays, and thus tighter achievable deadlines.

On the right, we sweep throughput and plot results for varying deadlines. Because CBH is computed independently of the deadline bound, the policy remains the same and is only plotted once. Once again, ARSC produces consistent results and low delays across throughputs and deadlines, while CBH sees much larger delays. 

\section{Conclusion}\label{sec:conclusion}

In this paper, we analyzed the impact of wireless interference on scheduling for service guarantees. We defined the feasible region of throughput and deadline guarantees for a solitary flow, and showed that without carefully structuring the order of a scheduling policy, packets can experience large delay. To alleviate this problem, we showed that under regular schedules, tight deadline guarantees can be made which are independent of the schedule length. Finally, we developed an algorithm to construct regular schedules in polynomial time, which are guaranteed to meet service requirements for all flows, while leaving sufficient capacity for best-effort traffic to achieve near-optimal throughput. Future work includes optimizing routing as well as scheduling, and dealing with unreliable links, changing wireless topologies, and stochastic traffic.

\appendix

\subsection{Proof of Theorem~\ref{th:cyclicschedules}}\label{app:cyclicschproof}

Recall that arrival rates are fixed on the interval $0 \leq t \leq T$, so exactly $\lambda_i$ packets of $f_i$ arrive in each slot. In order for a policy to support $\mathcal{F}$, no packet can exist in the system for longer than its deadline, so at most $\lambda_i \tau_i$ packets from $f_i$ can be present in the system at the end of a slot. This condition is both necessary and sufficient on this interval because packets are served in a FCFS manner, so $\pi$ supports $\mathcal{F}$ up to time $T$ if and only if $\sum_{e \in T^{(i)}} Q_{i,e}^{\pi}(t) \leq \lambda_i \tau_i$ for all $f_i \in \mathcal{F}$ on this interval. By assumption, at least one policy $\pi'$ exists that supports $\mathcal{F}$, so this condition must hold under $\pi'$.

Define the state of the system at time $t$ as the length of each queue and denote it by $s^{\pi'}(t)$. Because queue lengths are bounded under $\pi'$, there exist a finite number of states which can be visited over any time horizon $T$. In particular, for sufficiently large $T$, there must exist a state $s^{\pi'}(t_0)$ which occurs at time $t_0$ and then occurs again at time $t_0+K^{\pi'}$ for some $K^{\pi'} > 0$. Let $s^{\pi'}(t_0)$ be the first state where this event occurs, and let the set of actions taken in the time interval $[t_0,t_0+K^{\pi'})$ be the policy $\pi$ with $K^{\pi} = K^{\pi'}$ and $\mu^{\pi}(t) = \mu^{\pi'} \big((t-t_0) \bmod K^{\pi} + t_0 \big)$ for all $t \geq 0$. Then for all $t_0 \leq t \leq T$, we have $s^{\pi}(t) = s^{\pi'} \big((t-t_0) \bmod K^{\pi} + t_0 \big)$, so~\eqref{eq:queuebound} holds and $\pi$ supports $\mathcal{F}$ for all $t_0 \leq t \leq T$.

Note that $t_0 > 0$ because queues are empty at $t=0$ and the system requires time to ``ramp up''. However, we next show that $\pi$ supports $\mathcal{F}$ on the interval $0 \leq t < t_0$ using induction on the queue sizes. The queue evolution equations are given by 
\begin{equation}
    Q_{i,e}^{\pi}(t+1) = \min\{ Q_{i,e}^{\pi}(t) + \tilde{\lambda}_{i,e}(t) - \mu_e^{\pi}(t) w_{i,e}, \ 0 \},
\end{equation}

for all $0 \leq t \leq T$, where
\begin{equation}
    \tilde{\lambda}_{i,e}(t) = \begin{cases}
        \lambda_i, &e = T^{(i)}_0, \\
        \mu_{e_{-1}}^{\pi}(t) \cdot \min \{w_{i,e_{-1}}, \ Q_{i,e_{-1}}^{\pi}(t) \}, &e \neq T^{(i)}_0,
    \end{cases}
\end{equation}

and with a slight abuse of notation we denote the link preceding $e$ in $T^{(i)}$ as $e_{-1}$, where it is understood we are referring to $f_i$. Now assume that $Q_{i,e}^{\pi}(t) \leq Q_{i,e}^{\pi}(t+K^{\pi})$ for all $i$ and $e$. Then, from the queue evolution equations and given that $\mu^{\pi}(t) = \mu^{\pi}(t+K^{\pi})$, this implies that $Q_{i,e}^{\pi}(t+1) \leq Q_{i,e}^{\pi}(t+K^{\pi}+1)$ for all $i$ and $e$ as well. By definition, $Q_{i,e}^{\pi}(0) \leq Q_{i,e}^{\pi}(K^{\pi})$ for all $i$ and $e$, which completes the induction step. Therefore,
\begin{equation}
    \sum_{e \in T^{(i)}} Q_{i,e}^{\pi}(t) \leq \sum_{e \in T^{(i)}} Q_{i,e}^{\pi}(t+nK^{\pi}) \leq \lambda_i \tau_i
\end{equation}

for all $f_i$ and $0 \leq t < t_0$, and some $n>0$ such that $t_0 \leq t+nK^{\pi} < t_0 + K^{\pi}$, which shows that $\pi$ supports $\mathcal{F}$ on the interval $0 \leq~t < t_0$.

It can easily be shown that $\pi$ also supports $\mathcal{F}$ on the interval $t > T$. Assume for all $t > T$, dummy packets arrive with the same rate $\lambda_i$ for each flow, until all real packets have been delivered. Then, because $\pi$ supports every flow under regular arrivals, it must also support every flow under the dummy arrivals. This completes the proof.

\subsection{Proof of Theorem~\ref{th:orrbound}}\label{app:orrproof}

The ORR policy schedules links in order from source to destination in subsequent time slots according to the definition above. The activation rate follows because each link is only activated once per scheduling period. Now assume a packet arrives at the source at some time $t$, and that slice widths are large enough to serve all enqueued packets when a link is scheduled. This is non-restrictive because $\tau_i^*$ is defined as the smallest feasible deadline for an arbitrarily small throughput. At time $t \bmod{(\phi+1)}$, the packet is served at link $T^{(i)}_0$, and following the ORR schedule, it is served at each subsequent link in the next $|T^{(i)}|-1$ slots. In the worst case, the packet must wait $\phi$ slots at the source before being served, so it spends a maximum of $|T^{(i)}|+\phi$ slots in the network, which verifies $\tau_i^*$ for the ORR policy.

It remains to show that no other policy can achieve a smaller deadline for all packets. Recall that only one of $\{T^{(i)}_0, \dots, T^{(i)}_{\phi} \}$ can be scheduled in the same slot. We claim that it must take at least some packets $2\phi+1$ slots to reach link $T^{(i)}_{\phi+1}$. The fewest slots a packet can take is $\phi+1$, so we define $\Delta(t)$ as the number of additional slots it takes beyond this minimum for packets which arrive at the source at time $t$. Then if $\Delta(t) \geq \phi$ for any $t$, our claim must hold. Note that $\Delta(t)=0$ at time $t$ when packets arrive, and it is incremented by one each time a scheduling decision is made that does not schedule those packets.

Assume our claim does not hold, i.e., that $\Delta(t) < \phi$ for all $t$. First note that packets which arrive at $t$ can allow at most $\Delta(t)$ slots of arrivals behind them to ``catch up''. We say a slot of arrivals $t+j$ is caught up to $t$ if it is enqueued at the same link as the slot of arrivals $t$. Each slot of arrivals which catches up to $t$ increments $\Delta(t)$, so arrivals from slot $t+\phi$ cannot be caught up to $t$ under our assumption on $\Delta(t)$. Let $t+j^* \leq t+\phi$ be the first slot which is not caught up to $t$ when packets from slot $t$ reach $T^{(i)}_{\phi+1}$. Because the previous slot of arrivals is caught up to $t$, it must have been scheduled $\phi$ times independently of the arrivals from slot $t+j^*$. Therefore, $\Delta(t+j^*) \geq \phi$, which is a contradiction. This proves the result.

\subsection{Proof of Corollary~\ref{cor:jointoptimalpoint}}\label{app:jointoptproof}

From Theorem~\ref{th:orrbound}, any ORR policy is deadline-optimal, so it remains only to show that it is throughput-optimal. Under equal slice widths,~\eqref{eq:thoptformulation} becomes
\begin{align}
\begin{aligned}\label{eq:orrthopt}
    \max_{\pi \in \Pi_c} &\min_{e \in T^{(i)}} \bar{\mu}_e^{\pi} \\
    \text{s.t.} \ &\mu^{\pi}(t) \in M^{\phi} , \ \forall \ 0 \leq t \leq K^{\pi}.
\end{aligned}
\end{align}

Under $\phi$, only one of every $\phi+1$ consecutive slots in $T^{(i)}$ can be activated at once. If $|T^{(i)}| \leq \phi$, only one link can be scheduled at a time so let $\phi = |T^{(i)}|-1$ without loss of generality. Then,
\begin{equation}
    \sum_{e=j}^{j+\phi} \mu_e^{\pi}(t) \leq 1, \ \forall \ t \geq 0, \ 0 \leq j < |T^{(i)}| - \phi,
\end{equation}

and any $\pi \in \Pi_c$, where we slightly abuse notation to denote $T^{(i)}_j$ as $j$ in the subscript. Averaging these constraints over each scheduling period,
\begin{equation}
    \sum_{e=j}^{j+\phi} \bar{\mu}_e^{\pi} \leq 1, \ \forall \ 0 \leq j < |T^{(i)}| - \phi,
\end{equation}

and finally summing over these constraints,
\begin{multline}\label{eq:kphinecessarycond}
    \bar{\mu}_0 + \bar{\mu}_{-1} + 2(\bar{\mu}_1 + \bar{\mu}_{-2}) + \dots + \phi(\bar{\mu}_{\phi-1} + \bar{\mu}_{-\phi}) \\
    + (\phi+1) \sum_{j=\phi}^{|T^{(i)}|-\phi-1} \bar{\mu}_j \leq |T^{(i)}|-\phi, 
\end{multline}

where we again slightly abuse notation to denote $T^{(i)}_{|T^{(i)}|-j}$ as $-j$. Because the sum relaxes the constraints, this is a necessary but not sufficient condition for a feasible policy.

Counting terms, there are $(\phi+1)(|T^{(i)}|-\phi)$ activation rates summed in this expression. Under the ORR policy, each link $e$ has $\bar{\mu}_e^{ORR(i)} = \frac{1}{\phi+1}$, so summing over all of them shows that the bound in~\eqref{eq:kphinecessarycond} is tight under the ORR policy. Because the bound is a necessary condition for feasibility, any policy for which it is tight must be a solution to~\eqref{eq:orrthopt}. This can be shown because increasing $\bar{\mu}_e$ for any $e$ causes a decrease in $\bar{\mu}_{e'}$ for some $e'$, which strictly decreases the objective in~\eqref{eq:orrthopt}. Therefore, ORR must be throughput-optimal.

\subsection{Proof of Lemma~\ref{lemma:emptyqueues}}\label{app:emptyqueueproof}

Under resource-minimizing slices, $\bar{w}_{i,e}^{\pi} = \lambda_i$ for all $f_i \in \mathcal{F}$ and $e \in T^{(i)}$ by definition. Then the total number of $f_i$ packets served in each period is at most $\eta_e^{\pi} w_{i,e} = \bar{w}_{i,e}^{\pi} K^{\pi} = \lambda_i K^{\pi}$.

Every $e$ after the first link in $T^{(i)}$ sees at most $\lambda_i K^{\pi}$ arrivals per scheduling period, because this is the most packets which can be served at the previous link in the route. Furthermore, the first link in the route sees exactly this number of arrivals per scheduling period, so this holds for all links. Now assume that $Q_{i,e}(t) = 0$ at some time $t$, and let $t_{i,e}(t)$ be the first time after $t$ when any number of packets have arrived and been served, and the queue is empty again, so that $Q_{i,e}(t_{i,e}(t)) = 0$. Note that this always occurs in a slot immediately after link $e$ is served.

We claim that $t_{i,e}(t) \leq t + K^{\pi}$ for any $t$. To see this, recall that link $e$ serves at most $\eta_e^{\pi} w_{i,e}$ packets per period, but that any policy in $\Pi_c$ is work conserving, so it always serves the smaller of $w_{i,e}$ and its queue size when it is scheduled. Let $t'$ be the last time link $e$ is scheduled before time $t+K^{\pi}$, or equivalently the $\eta_e^{\pi}$-th time it is scheduled after time $t$. Then, either there are more than $w_{i,e}$ packets in link $e$'s queue the first $\eta_e^{\pi}-1$ times it is scheduled after time $t$, or the queue is emptied in one of these scheduling events, and under the primary interference assumption, no packets can be served at the previous link and arrive at link $e$ in the same slot. Therefore, in the latter case, $t_{i,e}(t)$ is the slot immediately following when it is emptied. In the former case, this implies that link $e$ serves $(\eta_e^{\pi}-1) w_{i,e}$ packets in the interval $[t,t')$, and because at most $\eta_e^{\pi} w_{i,e}$ packets can arrive during that interval, the queue must be emptied when it is scheduled at $t'$, and $t_{i,e}(t) = t'+1 \leq t + K^{\pi}$.

Now let $\tilde{t}=t_{i,e}(t)$, and by the same argument, there must exist a $t_{i,e}(\tilde{t}) \leq \tilde{t}+K^{\pi}$ where packets have arrived and been served, and $Q_{i,e}(t_{i,e}(\tilde{t})) = 0$. By induction, this holds for any scheduling period $K^{\pi}$ if $Q_{i,e}(t) = 0$ at some $t$. Because queues are initialized to zero, this completes the induction step, and the proof is complete. Note that after some $t_0$, when the network has reached steady state, $t_{i,e}$ will occur at regular intervals from Theorem~\ref{th:cyclicschedules}, but that these times are not guaranteed to coincide for different flows.

\subsection{Proof of Theorem~\ref{th:deadlinebound}}\label{app:deadlineboundproof}

Assume without loss of generality that slices are resource-minimizing. Allowing slice widths to be larger than this can only decrease the deadline bound, so this assumption is not only non-restrictive but necessary. From Lemma~\ref{lemma:emptyqueues}, under resource-minimizing slices there exists a slot for each flow in each scheduling period where $Q_{i,e}^{\pi}(t) = 0$. Without loss of generality, consider flow $f_i$ and let $t_{i,e}$ be any such time when $Q_{i,e}$ is empty. Under a fixed set of activation rates $\bar{\mu}^{\pi}$, link $e$ is activated exactly $\eta_e^{\pi} = \bar{\mu}_e^{\pi} K^{\pi}$ times per scheduling period. Therefore, it must be activated $\eta_e^{\pi}$ times in the interval $t_{i,e} \leq t < t_{i,e} +K^{\pi}$.

Assume that $t \geq t_0$, so the network has already reached steady-state and $Q_{i,e}(t_{i,e}+K^{\pi}) = Q_{i,e}(t_{i,e}) = 0$. This is non-restrictive because before $t_0$ there are fewer packets in each queue. One can simply add dummy packets to ensure that $Q_{i,e}(t) = Q_{i,e}(t+n_t K^{\pi})$ for some integer $n_t$ such that $t+n_t K^{\pi} \geq t_0$ for all $t < t_0$. Then in each scheduling period, exactly $\eta_e^{\pi} w_{i,e} = \lambda_i K^{\pi}$ packets arrive and are served at link $e$, so it always serves a full slice of packets. This can be shown by induction and the fact that $Q_{i,e}(t) = Q_{i,e}(t+K^{\pi})$. Because $\lambda_i K^{\pi}$ packets arrive at the source link in a scheduling period, the same number are served by the source link in this period. Therefore the same number arrive at the second link, and so the same number are served at the second link, and so on for each link. Because link $e$ is scheduled $\eta_e$ times and serves a full $w_{i,e}$ packets each time, the maximum delay that packets can experience at link $e$ occurs when link $e$ is scheduled the $\eta_e$ time slots leading up to time $t_{i,e}+K^{\pi}$, regardless of when packets arrived. 

Now consider the next link $e+1$ in the route. It also receives $\lambda_i K^{\pi}$ arrivals and must serve this number of packets in each scheduling period. As shown for link $e$, the maximum delay packets can experience occurs when it is scheduled in $\eta_{e+1}$ consecutive slots leading up to time $t_{i,e+1}+K^{\pi}$, where the $\lambda_i K^{\pi}$ packets from link $e$ arrive in the interval $t_{i,e+1} \leq t < t_{i,e+1}+K^{\pi}$. In particular, the worst-case delay occurs when link $e+1$ begins receiving packets from link $e$ at time $t_{i,e+1}+1$. The same holds for each link in $T^{(i)}$, and so the worst-case delay occurs when links are scheduled in blocks such that any link $e$ begins scheduling its block immediately after link $e+1$ finishes scheduling.

Because packets are scheduled in blocks, we consider the first and last packet in each block. Denote the time the first packet arrives at the source as $t(s_0)$ and the time the last packet in the block arrives at the source as $t(s_{-1}) = t(s_0) + K^{\pi}-1$. Similarly, denote the time the first packet is delivered to its destination as $t(d_0)$ and the time the last packet is delivered as $t(d_{-1}) \leq t(d_0) + K^{\pi}-1$. The total delay seen by the first packet is therefore larger than the last packet, and because all intermediate packets see a gradient of delay between these two values, the first packet sees the largest delay. It is always the first packet served in a block, so at each link $e$ it sees a delay of $K^{\pi}-\eta_e^{\pi} = K^{\pi}(1 - \bar{\mu}_e^{\pi})$ before being served. Finally, it takes one slot to be delivered to the destination after being served at the last link, so the bound follows.

This bound is tight for at least one policy by construction of the policy just described. We say that there are many policies where $\tau_i^*$ grows linearly with $K^{\pi}$, because any policy which schedules links in blocks incurs a delay at each link that is linear in $K^{\pi}$. In fact, from Lemma~\ref{lemma:taulowerbound}, any policy with inter-scheduling times that depend on the schedule length incurs such a delay.

\subsection{Proof of Theorem~\ref{th:generalcomplexity}}\label{app:complexityproof}

Following the same idea as in Theorem~\ref{th:singleflowcomp}, we can represent possible queue states of the network as nodes on a directed graph $G_{\mathcal{F}} = (V_{\mathcal{F}},E_{\mathcal{F}})$, with edges corresponding to valid activation sets, such that an edge $e$ from state $s_t$ to $s_{t+1}$ implies that activating the set of links corresponding to $e$ in state $s_t$ results in state $s_{t+1}$. The cost of each action is the sum of queue sizes at the destination node. Note that because slice widths are not fixed as in Theorem~\ref{th:singleflowcomp}, this allows for an infinite number of states corresponding to each action. To overcome this difficulty, we relax the condition 
\begin{equation}
    \sum_{f_i : e \in T^{(i)}} w_{i,e} \leq c_e, \ \forall \ e \in E
\end{equation}
so that each slice width $w_{i,e} = c_e$, for all flows $f_i$ and links $e \in T^{(i)}$. Then taking an action $e$ in state $s_t$ leads to a unique state, so the number of states is finite in the relaxed problem. Furthermore, any solution to the original problem is feasible for the relaxed problem, so the relaxed problem is at most as hard.

From Corollary~\ref{cor:deadlinecondition}, finding a policy that supports $\mathcal{F}$ is equivalent to finding a policy where $\sum_{e \in T{(i)}} Q_{i,e}(t) \leq \lambda_i \tau_i$ for all $t$ and for each flow $f_i$. Relaxing this condition gives 
\begin{equation}\label{eq:relaxedqueueconstr}
    \sum_{f_i \in \mathcal{F}} \sum_{e \in T^{(i)}} Q_{i,e}(t) \leq \sum_{f_i \in \mathcal{F}} \lambda_i \tau_i, \ \forall t \geq 0,
\end{equation}
where the left-hand side of~\eqref{eq:relaxedqueueconstr} is the cost of taking an action leading to state $s_t = \{ Q_{i,e}(t) \}$.

Finding a supporting policy is again equivalent to finding a cycle on $G_{\mathcal{F}}$ with bounded maximum cost, which can be further relaxed to finding a cycle with bounded average cost. Again, because of the relaxations, the original problem is at least as hard as this final relaxed problem, for which the best-known algorithm is $O(|V_{\mathcal{F}}||E_{\mathcal{F}}|)$ from Karp~\cite{karp1978characterization}.

Let $q_{i,e}$ be the number of discrete sizes that $Q_{i,e}$ can take, which is finite in the relaxed problem with fixed slices. Let $q \triangleq \min_{f_i \in \mathcal{F}, e \in T^{(i)}} q_{i,e}$, and use this to bound $|V_{\mathcal{F}}| \geq q^{\Psi}$, where $\Psi = \sum_{f_i \in \mathcal{F}} |T^{(i)}|$ is the total number of queues. There is at least one valid action in each state, so $|E_{\mathcal{F}}| \geq |V_{\mathcal{F}}|$, and we conclude that the original problem has complexity at least $O(q^{2 \Psi})$. This completes the proof.

\subsection{Proof of Theorem~\ref{th:gdmapproximation}}\label{app:gmapproxproof}

Assume $M^*$ is known, and let $\chi'_{M^*}$ be the number of matchings in $M^*$. Arrange the matchings in a histogram by their respective values of $\bar{\mu}_m$, from largest to smallest, assigning $\bar{\mu}_m$ to the height of column $m$. The resulting plot has $\chi'_{M^*}$ columns arranged in decreasing order and a total area equal to $\sum_{m \in M^*} \bar{\mu}_m$, which is precisely the optimal solution to~\eqref{eq:edgematchings}. Proceed in a similar way with $M^{GM}$. By definition, the matchings are already sorted by $\bar{\mu}_m$, so the resulting plot has $\chi'_{GM}$ columns arranged in decreasing order, and a total area equal to $\sum_{m \in M^{GM}} \bar{\mu}_m$.

Let $\bar{\mu}_{(i)}^*, \ i \in \{1,\dots,\chi'_{M^*} \}$ be the value of the $i$-th column of the (sorted) $M^*$, and $\bar{\mu}_{(i)}^{GM}, \ i \in \{1,\dots,\chi'_{GSM} \}$ be the equivalent under GM. Then we claim that
\begin{equation}\label{eq:histcolumns}
    \bar{\mu}_{(i)}^* \geq \bar{\mu}_{(i)}^{GM}, \ \forall \ i \in \{1,\dots,\chi'_{M^*} \},
\end{equation}

from the greediness of the GM algorithm. Assume that the optimal solution is also designed by an algorithm constructing super-matchings in order of sorted columns, and at each step $i$ let $E_{(i)}^*$ be the set of links/super-links not yet added to a matching. Likewise, let $E_{(i)}^{GM}$ be the equivalent for GM. Then, if
\begin{equation}
    \max_{e \in E_{(i)}^*} \bar{\mu}_e^0 \geq \max_{e \in E_{(i)}^{GM}} \bar{\mu}_e^0,
\end{equation}

is true at some step $i$, then by the greediness of GM it is also true at step $i+1$. Both algorithms start with the same set of links, so it must hold at $i=1$, which completes the induction step and proves that it holds for all $i \in \{1, \dots, \min \{\chi'_{M^*}, \chi'_{GM} \} \}$. Furthermore, because columns are sorted in decreasing order and every link/super-link is added to a super-matching, this implies that each algorithm adds the maximum remaining element in each step, so $\bar{\mu}_{(i)}^* \geq \bar{\mu}_{(i)}^{GM}$ for all $i \in \{1, \dots, \min \{\chi'_{M^*}, \chi'_{GM} \} \}$ as well. Finally, because the area under each histogram is by definition minimized by $M^*$, $\chi'_{M^*} \leq \chi'_{GM}$ and~\eqref{eq:histcolumns} holds.

Now imagine the horizontal axis of the histogram of $M^*$ is scaled by a factor of $\frac{|E|}{\chi'_{M^*}}$. The area scales by the same factor, and the plot now has nonzero values from $0$ to $|E|$ on the horizontal axis. Define this function as $H^*(e)$. Similarly, let the horizontal axis of the histogram of $M^{GM}$ be scaled, but this time by a factor of $\frac{|E|}{\chi'_{GM}}$. Once again, the area scales by the same factor and the plot has nonzero values from $0$ to $|E|$. Define this function as $H^{GM}(e)$.

Then $H^*(e) \geq H^{GM}(e)$ for all $e \in \{1,\dots,|E| \}$. This follows from~\eqref{eq:histcolumns}, the fact that $H^*(e)$ and $H^{GM}(e)$ are decreasing functions of $e$, and the fact that $\chi'_{M^*} \leq \chi'_{GM}$, so the scaled columns of $H^*$ are wider than those of $H^{GM}$. It follows that the total area under $H^*$ is no less than the area under $H^{GM}$, or more explicitly,
\begin{equation}
    \frac{|E|}{\chi'_{M^*}} \sum_{m \in M^*} \bar{\mu}_m \geq \frac{|E|}{\chi'_{GM}} \sum_{m \in M^{GM}} \bar{\mu}_m.
\end{equation}

Rearranging terms,
\begin{equation}
    \sum_{m \in M^{GM}} \bar{\mu}_m \leq \Big( \frac{\chi'_{GM}}{\chi'_{M^*}} \Big) \sum_{m \in M^*} \bar{\mu}_m,
\end{equation}

and recognizing that $\chi'_{M^*} \geq \chi'(G)$ for any $M^*$, and $\chi'(G) \geq \Delta(G)$ by Vizing's Theorem~\cite[Sec.~9.3]{gross2018graph}, the result follows.

Finally, note that the GM algorithm is nothing more than a greedy edge coloring on the graph $G$, where the edges are revealed to the algorithm in order of weight. It has been shown that greedy edge coloring uses at most $2 \Delta(G) - 1$ colors, independently of the order in which edges are revealed, provided that $\Delta(G) = O(\log |V|)$~\cite{bar-noy_greedy_1992}. This completes the proof.

\subsection{Proof of Lemma~\ref{lemma:regconditions}}\label{app:regconditionsproof}

We begin by showing the first part. Condition $(1)$ is necessary for any work-conserving schedule, and condition $(3)$ is likewise necessary for any regular schedule. It remains to show that in conjunction with condition $(2)$, they are sufficient conditions for a regular schedule, and we will prove this by construction.

By definition of a step-down vector, $\bar{\mu}_{|M|}$ is the smallest activation rate, and every other rate $\bar{\mu}_m = \alpha_m \bar{\mu}_{|M|}$, for some positive integer $\alpha_m$. To formalize this, let the the integer factor between nodes $i$ and $j$ where $i < j$ be $z_{ij}$, so that $\bar{\mu}_i = z_{ij} \bar{\mu}_j$. From condition $(3)$, $k_m = \frac{1}{\bar{\mu}_m}$ for all $m$, and after substituting and rearranging, we have
\begin{equation}\label{eq:stepdownlength}
    k_m z_{m |M|} = k_{|M|}, \ \forall \ m.
\end{equation}

Define the number of times matching $m$ appears in the schedule as $\eta_m = \bar{\mu}_m K$ for a schedule length $K$. Rearranging and substituting condition $(3)$ gives $k_m \eta_m = K$ for all $m$. Any $K$ which satisfies this equation for integer $\eta_m$ is a valid schedule length, so without loss of generality let $\eta_m = z_{m |M|}$ for all $m$, which gives $K = k_{|M|}$ from~\eqref{eq:stepdownlength}.

Define a schedule $\mathcal{S} = \{S_0, S_2, \dots, S_{K-1} \}$, where $S_i = m$ if matching $m$ is scheduled in slot $i$. Starting with matching $1$, schedule it in slots $\mathcal{S}_1 = \{ S_0, S_{k_1}, \dots, S_{(\eta_1-1) k_1} \}$. Note that this satisfies regularity, and that the starting position does not matter because the schedule is periodic, so without loss of generality we choose position $1$.


Now we proceed in order, beginning with matching $2$ and ending when all matchings are scheduled. By the step-down property, for matching $m$, we have $k_m = z_{lm} k_l$ for all $1 \leq l < m$. If an empty slot $\tau_m$ exists to schedule matching $m$, then subsequent slots at intervals of $k_m$ are also empty, and we set $\mathcal{S}_m = \{ S_{\tau_m}, S_{\tau_m+k_m}, \dots, S_{\tau_m + (\eta_m-1) k_m} \}$. We show this by contradiction, and define $\tau_1 = 0$ for ease of notation. Imagine one slot $\tau_m + c_m k_m$ is not empty, so it must belong to one of the previously scheduled matchings. Denote this matching as $m'$. Then $\tau_m + c_m k_m = \tau_m + c_m z_{m'm} k_{m'} = \tau_{m'} + c_{m'} k_{m'}$, which implies that $\tau_m = \tau_{m'} + \alpha k_{m'}$ for some integer $\alpha$, so it cannot be empty, which is a contradiction. Therefore, $\mathcal{S}_m$ is valid for any empty $\tau_m$, and without loss of generality we choose it to be the first empty slot.

As a result, scheduling matchings in order and assigning them to the first empty slot in the schedule (along with subsequent slots defined by their inter-scheduling times) guarantees a feasible regular schedule. Because $\sum_m \bar{\mu}_m = 1$, we have that $\sum_m \eta_m = K$, so the schedule will not run out of empty slots before scheduling all matchings. This proves the first part.

We now show that without the third condition, an almost-regular schedule still exists. We show this by constructing a longer, regular schedule and removing the unused slots. Let the elongated schedule length be $K' = k_1 \eta_1 = \lceil \frac{1}{\bar{\mu}_1} \rceil \eta_1$, and let $k_m' = K'/\eta_m$ for all $m$. Note that for matching $1$, $k_1 = k_1'$, and assign $\mathcal{S}_1 = \{ S_0, S_{k_1'}, \dots, S_{(\eta_1-1) k_1'} \}$. Proceed to schedule matchings $2$ through $m$ in nearly the same manner as before, but taking slightly more care in choosing $\tau_m$ for each $m$.

Rather than choosing the first empty slot for each value of $\tau_m$, we select it in the following manner. Denote the set of empty slots most closely following slots in $\mathcal{S}_1$ as $\chi_1$. In other words, given that $\mathcal{S}_1 = \{S_0, S_{k_1'}, \dots \}$, if any slots in $\{ S_1, S_{k_1'+1}, \dots \}$ are empty then assign them to $\chi_1$. If not, then if any slots in $\{ S_2, S_{k_1'+2}, \dots \}$ are empty then assign them to $\chi_1$, etc. When $\chi_1$ is no longer empty, find the subset $\chi_2 \subseteq \chi_1$ most closely following slots in $\mathcal{S}_2$, in the same fashion as before. Continue this process until finding $\chi_{m-1}$, and assign $\tau_m$ to be the first slot in $\chi_{m-1}$.

The proof of the first part of the Lemma holds for any choice of $\tau_m$, provided the slot is empty, so assigning $\tau_m$ as above also constructs a valid regular schedule of length $K'$. Because the schedule is elongated not all slots are filled, and the remainder of the slots are removed by compressing the schedule. Clearly the schedule is no longer regular after compression, but it is almost-regular by design.

To see this, consider the layout of ``holes'' left in the schedule after the addition of each matching. Following matching $1$, there are $\eta_1$ blocks of $k_1'-1$ holes, which are evently spaced by symmetry. The addition of subsequent matchings always fills the leftmost hole in a block, and always selects the largest remaining blocks to fill. This holds for $\tau_m$ by the iterative process of selecting $\chi_1, \chi_2, \dots$, and must also be true for the other holes filled by matching $m$ under the step-down property. If not, then there is some integer $c_m$ such that $\tau_m + c_m k_m' - 1 = \tau_m + c_m z_{\tilde{m} m} k_{\tilde{m}}' - 1 \neq \tau_{\tilde{m}} + c_{\tilde{m}} k_{\tilde{m}}'$, where without loss of generality we assume slot $\tau_m-1$ is assigned to $\tilde{m}$. From regularity, every slot $\tau_{\tilde{m}} + \alpha k_{\tilde{m}}'$ for any positive integer $\alpha$ is assigned to $\tilde{m}$, so this is a contradiction.

By always choosing the leftmost hole in a block, the number of blocks never increases and the rightmost hole in each block remains regularly spaced in the schedule. Furthermore, by always choosing the largest block, the size of any two blocks can never differ by more than one hole. Therefore, when all matchings have been scheduled, the number of holes that occur in an interval of $n$ slots cannot differ by more than one. Then because links are scheduled regularly, the number of holes between any two successive slots assigned to a matching $m$ cannot differ by more than one, so after removing the holes, the schedule must be almost-regular. This completes the proof.

\bibliographystyle{IEEEtran}
\bibliography{mobihoc_submission}


\end{document}